\documentclass[12pt]{article}

\usepackage{amsmath,amsthm,amsfonts,amscd,latexsym,amssymb,amsbsy,dsfont,mathrsfs,color}
\usepackage[small, centerlast, it]{caption}
\usepackage{epsfig}
\usepackage[titles]{tocloft}
\usepackage[latin1]{inputenc}
\usepackage{verbatim} 
\usepackage[active]{srcltx} 
\usepackage{fancyhdr}
\usepackage{caption}
\usepackage{enumerate}
\usepackage{color}
\usepackage{hyperref}

\definecolor{NiColor}{RGB}{77,77,255}
\definecolor{NiColoRed}{RGB}{255,77,77}
\definecolor{NiCitation}{RGB}{77,255,77}

\def\sp{\hskip -5pt}
\def\spa{\hskip -3pt}

\newsymbol\rest 1316    
\def\emptyset{\varnothing} 
\def\b1{{1\!\!1}}

\def\cB{\mathscr{B}}

\def\cD{\mathscr{D}}
\def\cF{{\ca F}}

\def\cH{{\ca H}}

\def\cL{\mathscr{L}}

\def\cS{\mathscr{S}}

\def\sA{{\mathsf A}}

\def\sM{{\mathsf M}}

\def\sO{{\mathsf O}}
\def\sH{{\mathsf H}}

\def\sP{{\mathsf P}}
\def\sQ{{\mathsf Q}}

\def\sV{\mathsf{V}}
\def\sT{{\mathsf T}}
\def\sE{{\mathsf E}}

\def\bC{{\mathbb C}}           

\def\bN{{\mathbb N}}
\def\bM{{\mathbb M}}

\def\bR{{\mathbb R}}

\def\bZ{{\mathbb Z}}

\def\gB{{\mathfrak B}}

\def\gE{{\mathfrak E}}
\def\gF{{\mathfrak F}}

\def\beq{\begin{eqnarray}}
\def\eeq{\end{eqnarray}}

\newcommand{\ca}[1]{{\cal #1}}         

\usepackage{sectsty}
\sectionfont{\normalsize}
\subsectionfont{\normalsize\normalfont\itshape}
\newtheoremstyle{TheoremStyle}
{3pt}
{3pt}
{\slshape}
{}
{\bf}
{:}
{.5em}
{}

\theoremstyle{TheoremStyle}
\newtheorem{theorem}{Theorem}
\newtheorem{corollary}[theorem]{Corollary}
\newtheorem{proposition}[theorem]{Proposition}
\newtheorem{lemma}[theorem]{Lemma}
\newtheorem{definition}[theorem]{Definition}
\newtheorem{remark}[theorem]{Remark}

\begin{document}


\hfill{\sl May  2023} 
\par 
\bigskip 
\par 
\rm


\par
\bigskip
\large
\noindent
{\bf  On the Relativistic Spatial Localization  for  massive real scalar Klein-Gordon quantum particles}
\bigskip
\par
\rm
\normalsize 


\noindent  {\bf Valter Moretti$^{a}$ }\\
\par

\noindent 
  Department of  Mathematics, University of Trento, and INFN-TIFPA \\
 via Sommarive 14, I-38123  Povo (Trento), Italy.\\
 $^a$valter.moretti@unitn.it\\

 \normalsize

\par

\rm\normalsize

\rm\normalsize


\par

\begin{abstract}
I rigorously analyze a proposal,  introduced by D.R.Terno, about a   spatial localization observable for a Klein-Gordon massive real particle  in terms of a Poincar\'e-covariant  family of POVMs. I prove that these POVMs are actually a kinematic deformation of the Newton-Wigner PVMs.  The first moment of one of these POVMs however  exactly  coincides  with a restriction (on a core) of the Newton-Wigner selfadjoint  position operator, though the second moment does not.
This fact permits to preserve all nice properties of the Newton-Wigner position observable, dropping the unphysical features arising from  the Hegerfeldt theorem. The considered POVM does not permit spatially sharply   localized states, but it admits families of almost   localized states with arbitrary precision. Next, I establish that the Terno localization observable satisfies part of a  requirement introduced by D.P.L.Castrigiano about causal temporal evolution concerning the Lebesgue measurable spatial regions of any Minkowskian reference frame. The validity of the complete  Castrigiano's causality requirement is also proved  for  a notion of spatial localization which generalizes Terno's one in a natural way.
\end{abstract}
\tableofcontents

\section{Introduction}

A long-standing puzzling  issue of theoretical and mathematical physics  concerns the notion of {\em spatial localization of  a relativistic particle at given time}. The problem is difficult because of  a number of no-go results popped out over the years,  after the seminal  work of Newton and Wigner \cite{NW}. 
These theoretical snags  establish  that  apparently natural proposals to define a spatial observable  of a relativistic free particle (for a given Minkowski reference frame at a certain time) are actually forbidden by  general requirements concerning causal locality and positivity of the energy. The first victim of these no-go results is the very  Newton-Wigner localization notion.

My opinion  is that this issue has been  quite overlooked in spite of being  urgent: after all, experimental physicists  can assert, with a certain approximation,   where a relativistic  particle has been detected at a certain time in laboratories.
{\em What theoretical notion describes these kinds of claims by our colleagues?}

 The notion of position observable not only is perfectly defined in the non-relativistic regime, but it plays a very  central role in the  theoretical construction of the {\em corpus} itself  of the quantum theory.  The notion of position is involved in the first version of the {\em canonical commutation relations} and the theoretical explanation of the {\em Heisenberg principle}.
{\em How  is it possible that a so crucial theoretical notion  simply fades out when we pass to the relativistic regime? }

The situation is very delicate from the physical perspective. First of all,  we know that, trying to localize a particle under its Compton length,  gives rise to a pair of particles, so that a sharp localization seems not possible. In this sense the detectors should play an {\em active role} \cite{tesi}. However that is a physical fact which is predicted by {\em interacting QFT}. It is not clear how such an  obstruction should take place in an elementary (perhaps naive) mathematical description that disregard  the effects of Quantum Field Theory.

 In author's view, however, the  intricate nature of the problem  is also due to a frequent  confusion in the literature concerning  two entangled,  but actually physically distinct  issues.

\begin{itemize} 
\item[(I1)]  On the one hand, one can  focus on the properties and theoretical assumptions on the probability  of spatial localization,  {\em without paying attention to the post-measurement state}.  In that case, the major obstructions against apparently natural definitions of localization observables arise from a class of theoretical  results cumulatively called {\em Hegerfeldt's theorem} \cite{Hegerfeldt,Hegerfeldt2} and their more advanced re-formulations \cite{Castrigiano1,Castrigiano2}. They at least   prove that no sharp localization is possible if the generator of time evolution is bounded below. Sharp localization would imply non local features of the (time evolution of the) position probability distributions: a  superluminal spread of the probability distribution \cite{Ruijsenaars}. 
These no-go results concern any general description of the spatial  localization observable (at a given time)  in terms of {\em positive-operator valued measures}  (POVMs) and not only  {\em projection-valued measures} (PVMs).  Castrigiano \cite{Castrigiano2}  formulated  a precise  causality condition
((b) in Definition \ref{LCR})
 that every physically acceptable POVM (or PVM) --  which describes spatial localization -- should  satisfy independently of the issue of the post-measurement state.
 
\item[(I2)]  On the other hand, one can (also) focus on the {\em post-measurement state} arising {\em after} a position measurement. In that case, a list of no-go results has been accumulated over the years starting from the so called {\em Malament theorem}. It  in particular establishes that localization cannot be described in terms of PVMs -- i.e., not even   in terms of self-adjoint operators. It happens   when (a) the post-measurement state is produced by a projective measurement,  (b) the PVM satisfies natural requirements of locality (according to Hellwig-Kraus' analysis \cite{HK}), and (c) the generator of time evolution is positive or bounded below. Reinforcing the hypotheses of Malament statement, the no-go result can be extended to localization observables in terms of POVMs as established first  by Busch \cite{Buschloc}  and by Halvorson and Clifton later \cite{HC},  when a suitable post measurement procedure has been chosen (essentially an ideal L\"uders measurement).
\end{itemize}

However, there is no automatic way to pass from (I1) to (I2), especially  when the position observable is described in terms of a POVM. There are infinitely many measurement schemes (based on completely positive maps) which give rise to the same PVM or POVM while the post measurement states are completely different.  This arbitrariness was already noticed by von Neumann in his seminal book on the mathematical foundations of Quantum Mechanics and it is a fundamental tool in the modern theory of quantum measurement \cite{Buschbook}. The fact that the values of a position observable are continuous is a further source of problems.  Continuity of outcomes   rules out all naive state-updating procedures  on account  of the crucial  {\em Ozawa theorem} \cite{Ozawa}. The standard projective L\"uders scheme is physically untenable in this case, even if it is always  described as the prototype of  all state updating processes in many  textbooks of quantum mechanics.

Referring to (I2), it seems to me  that these no-go results against {\em every} notion of spatial localization observable  always rely on a precise choice of the description of the post-measurement state in terms of Kraus operators. (E.g.,  they are the square root of the effects of the POVM).  In my opinion,  this choice appears to oscillate between being too naive or too arbitrary. Therefore some apparently definite claims, relying upon the issue (I2),  about the  non-existence of any spatial localization observable \cite{HC} do not seem really motivated up to now. Even if they impose some severe constraint on the measurement scheme, the last word has not been said in my view.

Both issues (I1) and (I2)  rule out in particular  the already cited  {\em Newton-Wigner position  observable} \cite{NW} of a quantum relativistic particle. 

The  Newton-Wigner position  observable is described in terms of a family of  PVMs  $\sQ_{n,t}=\sQ_{n,t}(\Delta)$ -- where $\Delta$ ranges in the measurable sets of the rest 3-space $\Sigma_{n,t}$
of every given Minkowski reference frame $n$ at every given time $t$. 
This family of PVMs is covariant with respect to the Poincar\'e group. It is worth stressing that covariance with respect to spatial Euclidean subgroup (and some further technical hypotheses) uniquely determine the   family of $\sQ_{n,t}$ as a consequence of Mackay imprimitivity theory as proved by Wightman \cite{Wightman}. This is one of the theoretical  motivations which make the NW position observable quite appealing.

In view of the spectral theorem, the information of the family of PVMs $\sQ_{n,t}$ is completely encapsulated in the assignment of a set of selfadjoint operators, the {\em Newton-Wigner position operators} 
$$ N_{n,t}^\alpha := \int_{\Sigma_{n,t}} x^\alpha d\sQ_{n,t}(x)\quad \alpha = 0,1,2,3\:,$$
 where  $x^0=t$ and the Minkowski coordinates (selfadjoint operators)  $N_{n,t}^1,N_{n,t}^2,N_{n,t}^3$ of a particle are  co-moving with $n$. Obviously $N^0_{n,t}=tI$.

To make more intricated the issue, the Newton-Wigner position selfadjoint operator $N^\mu_{n,t}$ possesses quite natural and appealing properties in spite of the fact that the associated PVM  violates basic local-causality principles. 
In particular (see Section \ref{NWsec}), explicitly referring to the case of a scalar massive particle:

(i) natural covariance properties with respect to the Lorentz   (and Poincar\'e) group take place (on a suitable domain):
 $$U_\Lambda N_{n,t}^\alpha U_\Lambda^{-1} = (\Lambda^{-1})^\alpha_\beta N^{\beta}_{\Lambda n, t_{\Lambda}}\:, \quad \forall \Lambda \in O(1,3)_+\:;$$

(ii) a quite natural relativistic version of {\em Ehrenfest's theorem} is valid for $k=1,2,3$:
$$ U^{(n)\dagger}_t  N_{n,0}^k U^{(n)}_t  =   N_{n,t}^k =   N_{n,0}^k   + t\frac{P_{nk}}{P_{n0}}\:;$$

(iii) the worldline determined by the expectation values $\langle \psi |N^\mu_{t,n} \psi \rangle$ is timelike as is expected by a massive particle:
$$ \sum_{k=1}^3 \left(\frac{d}{dt} \langle \psi|  N^{k}_{n,t} \psi\rangle\right)^2 < 1 \:;$$

 (iv) Heisenberg's commutation relations are satisfied on a suitable dense invariant  domain (a core)
 $$[N_{n,t}^k, N_{n,t}^h]=[  P_{n h}, P_{n k}]=0\:,  \qquad 
[ N_{n,t}^k, P_{n h}]= i\hbar\delta^k_hI \:,$$

(v) this in particular 
 produces the standard statement of the {\em Heisenberg principle};

(vi) when the energy content of a  state vector $\psi$ is small if compared with the $mc^2$ of the particle, then the $N^k_{n,0}\psi$ tends to become $X^k\psi$, where $X^k$ is the non-relativistic position operator.\\

This paper is devoted to address the issue (I1) for a scalar Klein-Gordon particle with mass $m>0$.  To this end, a recent proposal of (non-commutative) POVM localization observable $\sA_{n,t}(\Delta)$  will be considered for massive  spin-$0$ particles. This proposal was  due to Terno \cite{Terno}.
This notion of localization, contrarily to the Newton-Wigner notion of  localization does not admit sharply localized states (Proposition \ref{PROPNL}), so that it is not in automatic conflict with the Hegerfeldt theorem. 
However it admits states which resemble  localized states with arbitrarily fine approximation (Proposition \ref{PROPNL}).    
An idea of  proof that  the spatial decay of the Terno  probabilities does not trigger the Hegerfeldt's superluminal phenomena appears in \cite{Terno}. We shall rigorously prove this fact  as a byproduct of the  achievement (B) below.

We shall show (Theorem \ref{TEO0}) that the POVM $\sA_{n,t}(\Delta)$ is actually a kinematic deformation of the PVM $\sQ_{n,t}(\Delta)$ in terms of the components of the four-momentum $P_n^\mu$ in the used Minkowski reference frame $n$:
$$\sA_{t,n}(\Delta) = \sQ_{t,n}(\Delta) + \frac{1}{2}\left( \frac{P_{n\mu}}{P_{n0}}  \sQ_{n,t}(\Delta) \frac{P_{n}^\mu}{P_{n0}}  + \frac{m}{P_{n0}} \sQ_{n,t}(\Delta)  \frac{m}{P_{n0}} \right)\:.$$
This relation implies in particular that the family of POVMs $\sA_{n,t}$ satisfies a covariance property with respect to the Poincar\'e group analogous to the one satisfied by $\sQ_{t,n}$.

 Three main results are next achieved in this paper by expanding and making mathematically rigorous some definitions and results discussed  in \cite{Terno}
and referring to some  ideas introduced in \cite{Castrigiano2}.\

(A)   Theorem \ref{TEO1} proves that, in spite of the difference of the two POVMs,  the first-moment operator $X^\alpha_{n,t}$ of Terno's POVM  coincides with the Newton-Wigner position operator. Therefore  $X^\alpha_{n,t}$ preserves  all good properties (i)-(vi) of that operator listed  above but (v).  In fact, a corrected version of the Heisenberg inequality will be established
$$\Delta_\psi X^k_{n,t} \Delta_\psi P_{nk} \geq  \frac{\hbar}{2} \sqrt{1 + 2\Delta_\psi P_{n,k}^2 
\left\langle \psi \left|\frac{(P_{n0})^2-(P_{nk})^2}{(P_{n0})^{4}}\psi\right. \right\rangle}\:. $$
It evidently reproduces  the standard inequality for large values of the mass. 

(B) Theorem \ref{TEO2} proves  that Terno's notion of spatial localization  satisfies a consequence  of the causality  requirement introduced  by Castrigiano \cite{Castrigiano2} as conjectured by Terno \cite{Terno}.
The validity of this  condition rules out, in particular,  the obstruction represented by  the Hegerfeldt's theorem.

This pair of achievements  promote $\sA_{n,t}(\Delta)$ to be a very good candidate for the relativistic notion of spatial localization of a massive scalar particle from the viewpoint of the issue (I1) at least.

(C) The validity of the complete  Castrigiano causality requirement is finally established (Theorem \ref{LAST}). However this result needs an improved version of the family of POVMs $\sA$ and a delicate discussion about the physical nature  of spatial localization.

In the recent years, several interesting problems related to  the issue (I2) and  local causality have been fruitfully addressed in the setting of  {\em algebraic quantum field theory} by  Fewster,  Verch and collaborators \cite{FewsterVerch,BFR,FJR} in a given  curved (globally hyperbolic) spacetime.
These papers  complete and largely extend the fundamental analysis  by  Hellwig and Kraus \cite{HK}.
 In that case, the relevant notion of localization refers to spacetime regions and to generic local observables in the Haag-Kastler setting. This paper instead deals with single particles (not quantum fields) and the localization refers to the space of a reference frame at a given time in Minkowski spacetime. It is clear that this is an ideal description which perhaps will reveal unphysical eventually, since realistic measurements take a finite lapse of time necessarily.  However, up to now, this type of ideality does not seem a source of the above mentioned obstructions to the definition of  a physically meaningful notion of spatial localization. 
 On the other hand it seems remarkable the fact that the Terno notion of spatial localization is actually a byproduct of QFT, at least from  a heuristic perspective: it arises from the {\em normally-ordered stress-energy tensor operator} whose nature is intrinsically part of basic constructions of QFT.

 This paper is organized as follows. Section 2 contains a quick technical recap on the massive  Klein-Gordon field in Minkowski spacetime,  stressing in particular on the covariance properties with respect to the relevant Poincar\'e unitary representation. Section 3 introduces the Newton-Wigner notion of spatial localization according to Wightman viewpoint. Section 4 illustrates some well-known problems with the NW notion of localization also presenting  general Castrigiano's causality requirement and the notion of causal time evolution, proving that this notion of localization is ruled out by the Hegerfeldt theorem. Section 5 introduces the notion of spatial localization presented by Terno into a rigorous setting and establishes  some important properties of it. Section 6 proves that this notion of spatial localization is in agreement with Castrigiano's notion of causal time evolution. Section 7  focuses on the causality  condition proposed by Castrigiano by introducing a second family of POVMs depending on a pair of reference frames. The final section is devoted to a discussion on the achieved results and possible developments.

 \section{Minkowski spacetime and Klein-Gordon massive particles}
 \subsection{Minkowski spacetime}
 In the rest of the paper, the {\bf Minkowski spacetime} $\bM$ is described as a four-dimensional  real  affine space -- whose vector space of translations is denoted by $\sV$ -- endowed with a Lorentzian metric $g$ in $\sV$ with signature $-,+,+,+$. 
A basis $\{v_0,v_1, v_2,v_3\}\in \sV$ is said to be {\bf pseudo orthonormal} if $g(v_\mu,v_\nu) = \eta_{\mu\nu}$, where $[\eta_{\mu\nu}]= diag(-1,1,1,1)$.

{\bf Causal vectors} $v\in \sV$ satisfy per definition $g(v,v) \leq 0$ and $v\neq 0$. Causal vectors with $g(v,v)=0$ are said {\bf null} or {\bf lightlike}. They are {\bf timelike} if $g(v,v) <0$. Finally,
{\bf spacelike vectors} satisfy $g(v,v)>0$.

 $(\bM,g)$ is {\bf time-oriented},  i.e., we choose a preferred half 
$\sV_+$ of the open cone of the  timelike vectors,   $g(v,v)<0$. The (causal!) vectors in $\overline{V_+}\setminus \{0\}$ are said {\bf future-directed}.  $\sT_+:=\{ v\in \sV_+ \:|\: g(v,v) = -1\}$ is the  set of unit future-directed timelike vectors.  The remaining half of the open cone $\sV$ of timelike vectors is denoted by  $\sV_-$.   The  {\bf past-directed} causal vectors are the elements of 
 $\overline{\sV_-}\setminus \{0\}$. The past-directed timelike and lightlike vectors are analogously the elements of $\sV_-$ and $\partial \sV_-\setminus \{0\}$ respectively.

 $J^+(S) \subset \bM$ denotes the {\bf causal future} of $S\subset \bM$.  It is the set of events $e\in \bM$ such that there is some $e'\in S$ such that
$e-e' \in \overline{\sV_+}$. An analogous definition is valid for the {\bf causal past} $J^-(S)$ of $S$. Notice that $S\subset J^{\pm}(S)$, $A\subset B$
implies $J^\pm(A) \subset J^\pm(B)$,
 and  $J^\pm\left(\bigcup_{\alpha \in A}S_\alpha\right)= \bigcup_{\alpha\in A}J^\pm(S_\alpha)$.\\


\begin{remark}
Throughout  $v\cdot u := g(v,u)$ when $u,v\in \sV$.
The light speed  is $c=1$ and the Planck constant satisfies $\hbar = 1$  unless I will  specify otherwise. \hfill  $\blacksquare$
 \end{remark}  
  
\subsection{Poincar\'e group, reference frames, and all that}
I adopt the conventions of \cite{FKW} regarding the interpretation of the relevant groups of transformations in $\bM$.
The {\bf orthochronous Lorentz group} $O(1,3)_+$ is the group of {\em linear} maps $\Lambda : \sV \to \sV$ which both preserve the metric $g$ and $\sV_+$.  The {\bf orthochronous Poincar\'e group} $IO(1,3)_+$ is the group of {\em affine} maps $\bM\to \bM$ whose associated  linear map belongs to  $O(1,3)_+$.  

If $A\subset \bM$ and $h\in O(1,3)_+$, then  $hA := \{h(e)\:|\: e \in A\}.$

 Every $n\in \sT_+$ defines a corresponding (Minkowskian) {\bf reference frame} in $\bM$. The three-dimensional {\bf rest spaces} of the reference frame $n$
 are the three-planes (pseudo ortho) normal to $n$. To label them,  one chooses a preferred point $o \in \bM$ called {\bf origin}. (Everything is discussed in this paper does not depend on this choice.)  
 A rest space of $n \in \sT _+$ is therefore denoted  by $\Sigma_{n,t}$, where $t\in \bR$ indicates the signed distance (the proper time of $n$)
of $\Sigma_{n,t}$ from $o$:
\beq \Sigma_{n,t} := \{e\in \bM \:|\: -(e-o)\cdot n = t\}\:.\label{Sigma}\eeq

With a choice of the origin $o\in \bM$, the orthochronous Poincar\'e group $IO(1,3)_+$ is isomorphic to the semidirect product of $O(1,3)_+$ and $\sV$ itself and acts as follows \beq (\Lambda, a) : \bM \ni e \mapsto o+ a+ \Lambda (e-o) \in \bM\quad \mbox{for $(\Lambda,a) \in O(1,3)_+\times \sV$}\:.\label{rapPG}\eeq

By construction,  if $h:=(\Lambda_h,a_h) \in O(1,3)_+$, \beq h \Sigma_{n,t} = \Sigma_{\Lambda_h n, t_{h}}\quad \mbox{where $t_{h} := -(he -o) \cdot  \Lambda_h n$ for every $e\in \Sigma_{n,t}$.}
\eeq
Notice that it turns out that $t_{h} = t- a \cdot \Lambda_h n$ does not depend on  the choice of $e\in \Sigma_{n,t}$.

The {\bf Euclidean group} ${\cal E}_n$ of $\Sigma_{n,t}$, i.e., the group of $h_{n,t}$-isometries, coincides with the subgroup of $IO(1,3)$ of elements $(\Lambda,a)$, which preserve $n$:
\beq
{\cal E}_n := \{ h\in O(1,3)_+ \:|\:  \Lambda_h n= n \} \label{EN}\:.
\eeq

With the choice of an origin $o$, $\bM$ is identified to $\sV$ by means of the bijective  map $M \ni e \mapsto e-o \in \sV$.  The choice of a basis  $\{v_1,\ldots, v_4\}\subset \sV$ defines a (global)  {\bf Cartesian coordinate system} of origin $o$ given by
$\bM \ni e \mapsto (x^1(e),\ldots, x^4(e)) \in \bR^4$ where  $e= o+ \sum_{\alpha=1}^4 x^\alpha(e)v_\alpha$.  That system of Cartesian coordinates is said to be {\bf Minkowskian} if the basis is pseudo orthonormal.
A {\bf Minkowskian  coordinate system}, 
with coordinates  $x^0=t,x^1,x^2,x^3$, is {\bf co-moving with $n\in O(1,3)_+$} if $\frac{\partial}{\partial x^0}=n$. Evidently $x^1,x^2,x^3$ define (global) Cartesian  orthonormal coordinates on each $\Sigma_{n,t}$ referring to the Euclidean  metric $h_{n,t}$ induced on it by $g$.

$\cB(\Sigma_{n,t})$ will denote the family of Borel subsets on $\Sigma_{n,t}$. 
Independently of the choice of the coordinates, $h_{n,t}$  induces a positive regular Borel measure  $d\Sigma_{n,t}$ on $\Sigma_{n,t}$. In  the above coordinates $x^1,x^2,x^3$, that measure is the restriction $d^3x=dx^1dx^2dx^3$   of  the Lebesgue measure  on $\bR^3$   to the Borel sets. The completion of $d^3x$  is the Lebesgue measure  itself as a consequence. The corresponding  completion of $d\Sigma_{n,t}$ will be named  {\bf Lebesgue measure} on $\Sigma_{n,t}$. I will make use of the same symbol $d\Sigma_{n,t}$ for a measure and its completion as the difference will be clear from the choice of the used $\sigma$-algebra.
The Lebesgue $\sigma$-algebra on $\Sigma_{n,t}$ will be denoted by $\cL(\Sigma_{n,t})$. 
\subsection{Completion of measures and $L^2$ spaces}\label{L2}
A positive $\sigma$-additive measure $\mu : \Sigma(X) \to [0,+\infty]$ and its completion  $\overline{\mu} : \overline{\Sigma(X)} \to [0,+\infty]$ give rise to the same Hilbert space $L^2(X, \mu)$ since  (see e.g., Proposition 1.57 \cite{Moretti1}), for every 
$\overline{\Sigma(X)} $-measurable function $f$, there is a $\Sigma(X)$-measurable function $g$ such that $f=g$ is true $\mu$-almost everywhere
 and either $\int_X f d\overline{\mu}= \int_X g d\mu$ or both  the integrals do not exist.  The identity evidently  extends  to $L^2$-scalar products of pairs of corresponding functions.  The map  $L^2(\mu) \ni [f]_\mu \mapsto [f]_{\overline{\mu}} \in L^2(\overline{\mu})$
is a Hilbert space isomorphism.

\subsection{Hilbert space and Poincar\'e group representation for the massive Klein-Gordon particle}
In the rest of this work,  I will take advantage of the Einstein convention of summation  over repeated Greek indices, from $0$ to $3$.

Let us consider a Klein-Gordon real particle of mass $m>0$ described by the $C^\infty$ scalar field $\varphi: \bM \to \bR$ satisfying the normally hyperbolic {\bf Klein-Gordon equation}
$$\Box \varphi - m^2 \varphi=0\:, \quad \mbox{where $\Box :=\eta^{\mu\nu} \partial_\mu\partial_\nu$ in every Minkowski coordinate system\:.}$$
As is well-known the quantization of that system, viewed as the restriction to the one-particle space of the second quantization procedure, relies  on the Hilbert space  of pure state vectors
$$\cH := L^2(\sV_{m,+}, \mu_m)\:.$$
Above, if $\sV_{m,+} := \{p\in \sV \:|\: g(p,p) = -m^2\:, \:p\in \sV_+\} $ denotes the {\bf mass shell} of (positive energy) four-momenta of mass $m$,  the Hilbert space inner product reads
\beq \langle \psi| \psi'\rangle := \int_{\sV_{m,+}} \overline{\psi(p)}\psi'(p) d\mu_m(p) \label{PS}\:. \eeq
Above, $\mu_m(p)$ is the Lorentz-invariant (positive Borel regular)  measure which takes the form 
\beq d\mu_m(p) = \frac{d^3p}{E_n(p)} \:, \quad E_n(p):= -n\cdot p\eeq in every Minkowskian reference frame co-moving with $n\in  \sT_+$, $d^3p= dp^1dp^2dp^3$ being the standard Lebesgue measure on $\bR^3$ identified with the rest spaces of $n$ by means of any Minkowskian coordinate system co-moving  with $n$ (that measure  is independent of the chosen Minkowskian coordinate frame co-moving with $n$).
Notice that 
$$E_n(p) = \sqrt{\vec{p}_n^2 + m^2} =p^0\:,  \quad \vec{p}_n:=  p + (n \cdot p)n \equiv  (p^1,p^2,p^3)$$
are respectively  the  $n$-temporal component and  $n$-spatial component  of the four-momentum $p$ 
respectively corresponding to $p^0$ and the triple $(p^1,p^2,p^3)$ in  any Minkowski coordinate system co-moving  with $n$.
As $E_n(p)$ depends only on $\vec{p}_n$, I will occasionally write $E_p(\vec{p}_n)$ in place of $E_n(p)$.

As usual, the (normal pure) {\bf quantum states} of the particle are represented by  the unit vectors $\psi \in {\cal H}$ up to phases.

 The inner product (\ref{PS}) is invariant under the strongly-continuous unitary  (active) action induced by\footnote{It is easy to prove that the result does not depend on the choice of $o$.}
(\ref{rapPG})
of the orthochronous Poincar\'e group $IO(1,3)_+$:
\beq (U_{(\Lambda, a)}\psi)(p) := e^{-i p\cdot a} \psi(\Lambda^{-1}p) \quad \mbox{if $\psi\in {\cal H}$ and  $(\Lambda, a)\in IO(1,3)_+$}\label{ACT1}\eeq This invariance property arises from the $O(1,3)_+$ invariance of  $\mu_m$: 
\beq \mu_m(\Lambda E) = \mu_m(E) \quad \mbox{for every Borel set $E$ in $\sV_{m,+}$.}\label{invM}\:.\eeq

The action of {\bf time translations} subgroup  along the time direction $n\in \sT_+$ reads $$U_{(I, \tau n)}\psi(p) = e^{i \tau E_n(p)}\psi(p)\:,$$
so that the self-adjoint generator of the one-parameter group, the multiplicative operator  \beq (P_{n 0}\psi)(p) := -(H_n\psi)(p):= -E_n(p)\psi(p)\eeq 
$$  D(H_n):= \left\{\psi \in L^2(\sV_{m+}, d\mu_m)\:\left|\: \int_{\sV_{m,+}} E_n(p)^2 |\psi(p)|^2 d\mu_m <+\infty\right. \right\}$$  has negative spectrum
since $\sigma(H_n) = \sigma_c(H_n) = [m,+\infty)$.  In this formalism, the {\bf time evolutor} in $n$
 is \beq U^{(n)}_\tau := U_{(I, -\tau n)} = e^{-i\tau H_n}\:. \label{TE} \eeq
 $H_n$ is the {\bf Hamiltonian operator} in the reference frame $n\in  \sT_+$.
The selfadjoint generators of the spatial  translations
 $$U_{(I,  a v_k)}\psi(p) = e^{-i a p_k}\psi(p)\:,$$
 in $n$ along the spatial unit vectors  $v_k$  of a co-moving Minkowskian coordinate system are therefore the multiplicative operators
 \beq P_{n k} :=  p_k \cdot = \vec{p}_k \cdot,\quad  k=1,2,3.\eeq
$$  D(P_{n k}):= \left\{\psi \in L^2(\sV_{m+}, d\mu_m)\:\left|\: \int_{\sV_{m,+}} (\vec{p}_n)_k^2 |\psi(p)|^2 d\mu_m <+\infty\right. \right\}.$$
Evidently $\sigma(P_{nk})= \sigma_c(P_{nk})= \bR$ for $k=1,2,3$.\\
The operators $(P_{n0}, P_{n1}, P_{n2}, P_{n3})$ define the (covariant) {\bf  components of the four-momentum} in $n$ with respect to the relevant Minkowskian coordinate system co-moving with $n$.  No specification of time $t$ is necessary because $P_{n \alpha}$ is trivially a constant of motion.\\
 
\begin{definition}
We say that $\psi \in \cal H$ is of {\bf Schwartz type} if  there is $n\in  \sT_+$ and  a Minkowski coordinate system co-moving  with $n$ such that 
$\bR^3 \ni \vec{p} \mapsto \psi(E_n(p), \vec{p}_n)\in \bC$ stays in ${\cS}(\bR^3)$ (the Schwartz space on $\bR^3$) when  represented in the spatial coordinates on $\bR^3$. 
The $\cal H$ subspace of vectors of Schwartz type will be denoted by ${\cal S}({\cal H})$.\\
\end{definition}

\begin{proposition}\label{PROPS}  The definition of ${\cal S}({\cal H})$ does not depend of the choice of  $n$ and co-moving Minkowskian coordinates. That is equivalent to saying the   ${\cal S}({\cal H})$ is  invariant under the representation $U$ of $IO(1,3)_+$ in (\ref{ACT1}). Finally, ${\cal S}({\cal H})$  is  dense in ${\cal H}$.
\end{proposition}
\begin{proof} See Appendix \ref{APPPROOFS}
\end{proof}

\begin{proposition}\label{PA}  ${\cal S}({\cal H})$ is invariant under the components of the four-momentum $P_{n\alpha}$, $\alpha=0,1,2,3$, referred to a reference frame $n\in  \sT_+$. Furthermore,  ${\cal S}({\cal H})$ is a core for each those symmetric operators (i.e., each of them is  essentially selfadjoint thereon).
\end{proposition}
\begin{proof} See Appendix \ref{APPPROOFS}
\end{proof}

 If $\psi \in {\cal S}({\cal H})$, the associated {\bf covariant  wavefunction} (the name is justified by (\ref{COVp})  below) is
\beq
\varphi_\psi(x) :=  \int_{\sV_{m,+}}  \frac{\psi(p)}{(2\pi)^{3/2}} e^{i p\cdot x} d\mu_m(p)\:, \label{wave}
\eeq
where $x(e) = e-o\in \sV$ is the vector representation of the events in $\bM$ with respect to the origin $o$. 
\begin{proposition}
If $\psi \in {\cal S}({\cal H})$, the associated wavefunction $\varphi_\psi$ satisfies the following.
\begin{itemize}
\item[(1)]  $\varphi_\psi\in C^\infty(M; \bC)$ and $\varphi_\psi(t,\cdot)  \in \cS(\bR^3)$ for every $t\in \bR$,  where $\bR^3 \equiv \Sigma_{n,t}$ through the  choice of a Minkowskian coordinate system co-moving  with any chosen $n\in  \sT_+$. 
\item[(2)] The Klein-Gordon equation is valid,
$\Box \varphi_\psi - m^2 \varphi_\psi=0\:.$
\item[(3)] If also $\psi' \in {\cal S}({\cal H})$, then
\beq\langle \psi |\psi' \rangle = \frac{i}{2}\int_{\Sigma_{n,t}}  \left(\overline{\varphi_\psi} \partial_n \varphi_{\psi'} -   \overline{\varphi_{\psi'}}  \partial_n \varphi_\psi \right) d\Sigma_{n,t} \label{SP}\eeq
where the  the right-hand side does not depend on the choice of both  $n\in  \sT_+$ and $t\in \bR$ since the left-hand side does not. 
\item[(4)]  The action (\ref{ACT1}) of  $IO(1,3)_+$ induces the standard active action on scalar fields in $\bM$,
\beq
\varphi_{U_{(\Lambda, a)}\psi}(x) = \varphi_\psi\left(\Lambda^{-1}(x-a)\right)\:. \label{COVp}
\eeq
\end{itemize}
Finally,  ${\cal H}$  coincides with  the completion of  ${\cal S}({\cal H})$  equipped with the inner product provided by the right-hand side of  (\ref{SP}). \\
\end{proposition}
I leave the proof of these very well known facts  to the reader. They are  based on elementary results of the theory of Fourier(-Plancherel) transform. The last statement immediately arises from (\ref{SP}) and the last statement of  Proposition \ref{PROPS}.

\section{The Newton-Wigner  observable for the massive Klein-Gordon  particle}\label{NWsec}

\subsection{The Newton-Wigner PVM}

I assume that the reader is well acquainted with basic notions of spectral theory and the notion of {\bf Projection Valued Measure} (PVM) (see, e.g., \cite{Moretti1,Moretti2}). 

 Consider a (separable) Hilbert space ${\cal H}$ that defines  the pure states of a quantum particle, not necessarily Klein-Gordon nor relativistic, but possibly  equipped with spin and other internal observables.
According to Wightman \cite{Wightman}, \\

\begin{definition}
A {\bf Newton-Wigner PVM} \cite{NW,Wightman}  for a  particle described in the (complex, separable) Hilbert space ${\cal H}$ is defined as a PVM $\sP: \cB(\bR^3) \to \gB({\cal H})$ -- where $\cB(\bR^3) $ is the Borel $\sigma$-algebra of $\bR^3$ -- which is {\bf covariant} with respect to a strongly continuous unitary representation $V$ of the group of isometries ${\cal E}$ of $\bR^3$ in ${\cal H}$: 
\beq
V_g \sP(\Delta) V_{g}^{-1} = \sP(g\Delta)\:, \quad \forall \Delta \in \cB(\bR^3) \:, \:\forall g \in {\cal E}\:. \label{COV}
\eeq
$\bR^3$ is above interpreted as the joint spectrum of three {\bf Newton-Wigner position  selfadjoint operators}
\beq
R_k := \int_{\bR^3} x_k d\sP(x_1,x_2,x_3)\:, \quad k=1,2,3.
\eeq
\end{definition}
\begin{remark}
{\em Wightman,  on a account of Mackey's {\em imprimitivity systems  theory},
established the uniqueness of  a Newton-Wigner position observable of a given unitary and strongly-continuous representation $V$ of the Euclidean group ${\cal E}$ 
under suitably regularity requirements on  $V$ and invariance under time-reversal symmetry. A more recent   discussion appears in \cite{Castrigiano2}. For a technically extensive discussion concerning relativistic systems  with every value of the square mass (also understood as an operator)  and the  spin see \cite{Castrigiano1,Castrigiano2}.}  \hfill $\blacksquare$\\
\end{remark}

According to the general interpretation of the formalism, the physical interpretation of a Newton-Wigner PVM is that $\langle \psi| \sP(\Delta)\psi \rangle$ is the probability to find the particle in the region $\Delta \subset \bR^3$ when the pure  state is represented  by $\psi \in {\cal H}$.

In the case of the real scalar Klein-Gordon particle, a Newton-Wigner PVM\footnote{which eventually can be proved to be unique on account of Wightman uniqueness theorem above mentioned \cite{Wightman}.} $\sQ_{n,t}$ is constructed as follows on the rest 3-space $\Sigma_{n,t}$
 of a reference frame $n\in  \sT_+$. Here,  the restriction $V$ of $U: IO(1,3)_+ \to \gB({\cal H})$ (\ref{ACT1}) to the Euclidean subgroup ${\cal E}_n$  (\ref{EN}) is used  to implement Wightman's definition.
As before,  events $e\in \bM$ are identified with vectors through $x(e) = e-0 \in \sV$. 

 If $n\in \sT_+$, $t\in \bR$, and  $\psi \in  {\cal S}({\cal H})$ define 
\beq \left(\sQ_{n,t} (\Delta)\psi\right)(p)  :=\sp\int_\Delta \sp d\Sigma_{n,t}(x)  \int_{\sV_{m,+}} \sp\sp\sp\sp  d\mu_m(q) \frac{e^{-i(p-q)\cdot x}}{(2\pi)^3}  \sqrt{E_n(p)E_n(q)} \psi(q)\quad \mbox{with $-n \cdot x=t$.}\label{firstPVM}\eeq
Above $\Delta \in \cB(\Sigma_{n,t})$ and $d\Sigma_{n,t}(x)= d^3x$ in Minkowskian coordinates co-moving  with $n$.
As  the mathematical tools appearing in the formula are coordinate independent for a choice of $n\in  \sT_+$, the operator on the  left-hand side only depends on  $(n,t)$.  The found family of operators defines  a Newton-Wigner observable on every slice $\Sigma_{n,t}$ according to Wigner's definition because of the following result.\\

\begin{proposition}\label{FPROP}
Each  operator of the
$(n,t, \Delta)$-parametrized
 family  (\ref{firstPVM}) defined on  ${\cal S}({\cal H})$ and taking values in $\cal H$, uniquely extends by continuity to the whole space $\cal H$.  
The found family of operators, for $t\in \bR$ fixed,
  defines  a PVM on  $\cB(\Sigma_{n,t})$ satisfying the covariance requirement (\ref{COV}) with respect to the group of isometries ${\cal E}_n$ (\ref{EN})  of $\Sigma_{n,t}$. \\
If indicating the found   orthogonal projectors   with the same symbol $\sQ_{n,t}(\Delta)$,
the action of  $IO(1,3)_+$  on them reads
\beq
  U_{h}  \sQ_{n,t}(\Delta) U_{h}^{-1} = \sQ_{\Lambda_h n,t_h}(h\Delta) \:, \quad \forall \Delta \in \cB(\Sigma_{n,t})\:, \quad h \in IO(1,3)_+\label{COVNW}\:.
\eeq
\end{proposition}

\begin{proof} Fix a Minkowskian coordinate system co-moving with $n$.  Define the unitary map \beq S_n: L^2(\sV_{m,+}, \mu_m) \ni \psi(p) \mapsto \frac{\psi(E_n(p), \vec{p}_n)}{\sqrt{E_n(p)}} \in L^2(\bR^3, d^3p)\:,\label{mapV}\eeq
where $\vec{p}_n\equiv  (p_1,p_2,p_3) \in \bR^3$ according to  the said choice of a Minkowskian coordinate system. Notice that, as $m>0$, the written map restricts to a bijection
from ${\cal S}({\cal H})$, which is dense in ${\cal H}=L^2(\sV_{m,+}, \mu_m)$, onto $\cS(\bR^3)$  viewed as dense subspace of $L^2(\bR^3, d^3p)$. 
We then  have that, for $\psi \in {\cal S}({\cal H})$,  (\ref{firstPVM}) can be reformulated as
\beq  \label{CONVNW2}   \sQ_{n,t} (\Delta)\psi =  U_{(I, t n)} S_n^{-1} \cF \:1_\Delta\: \cF^{-1}   S_n U^{-1}_{(I, t n)}\:\psi\eeq
Above, $1_\Delta$ is the multiplicative operator with the characteristic function of $\Delta \in \bR^3 \equiv 
\Sigma_{n,t}$ ($1_\Delta({x})=1$ if ${x}\in \Delta$ and $1_\Delta({x})=0$ otherwise);
 $\cF : L^2(\Sigma_{n,t}, d\Sigma_{n,t}) \to L^2(\bR^3, d^3p)$ is the Fourier-Plancherel unitary transform (after having identified $\Sigma_{n,t}$ with $\bR^3$ and  $d\Sigma_{n,t}$ with the Lebesgue measure $d^3x$ with the same  a choice of a Minkowskian coordinate system as above). 
$\cF$ and its inverse preserve the Schwartz space.
 The map $\cB(\bR^3) \ni \Delta \mapsto 1_\Delta \in \gB(L^2(\bR^3,d^3x))$ is evidently a PVM in the written Hilbert space.  As  $\cF^{-1} S_n$ is  norm preserving,  and when restricted to the dense subspace  of Schwartz functions has  a dense  range, $ S_n^{-1} \cF\: 1_\Delta\: \cF^{-1}  S_n|_{{\cal S}({\cal H})}$, extends to a bounded operator everywhere defined which is also a PVM.
Identity (\ref{COV}) is an immediate consequence  of (\ref{COVNW}) when ${\cal E}_3$ is identified with ${\cal E}_n$ (\ref{EN}).  Let us prove  (\ref{COVNW}). From (\ref{CONVNW2}), for $\psi ,\psi'\in {\cal S}({\cal H})$, the Fubini and Tonelli theorems yield
$$\langle \psi'| \sQ_{n,t}(\Delta) \psi \rangle =$$
$$  \int_{\sV_{m,+}} \sp\sp\sp\sp  d\mu_m(p) \overline{\psi(p)'}\int_\Delta \sp d\Sigma_{n,t}(x)  \int_{\sV_{m,+}} \sp\sp\sp\sp  d\mu_m(q) \frac{e^{-i(p-q)\cdot x}}{(2\pi)^3}  \sqrt{E_n(p)E_n(q)} \psi(q)$$
\beq 
 =\int_{\Sigma_{n,t}} \sp\sp\sp d\Sigma_{n,t}(x) 1_{\Delta}(x)\int_{\sV_{m,+}} \sp\sp\sp d\mu_m(p) \frac{e^{-i p\cdot x} \sqrt{E_n(p)}}{(2\pi)^{3/2}} \overline{\psi'(p)}   \int_{\sV_{m,+}} \sp\sp\sp d\mu_m(q) \frac{e^{iq\cdot x}\sqrt{E_n(q)}}{(2\pi)^{3/2}} \psi(q) \label{interm}
\eeq
where $-n\cdot x= t$ and the integrals are interpreted in proper sense. 
Let us define 
$$f_n(x):=\int_{\sV_{m,+}} \sp\sp\sp d\mu_m(p) \frac{e^{-i p\cdot x} \sqrt{E_n(p)}}{(2\pi)^{3/2}} \overline{\psi'(p)}   \int_{\sV_{m,+}} \sp\sp\sp d\mu_m(q) \frac{e^{iq\cdot x}\sqrt{E_n(q)}}{(2\pi)^{3/2}} \psi(q)\:.$$
At this juncture, taking advantage of (\ref{invM}) and
observing that the $\bM$-isometry  invariance   of the measures induced by the metric $d\Sigma_{n,t}(x)  = dh\Sigma_{n,t}(hx) $ entails, for $h\in IO(1,3)_+$
$$ \int_{x\in \Sigma_{n,t}} \sp\sp\sp d\Sigma_{n,t}(x) 1_{\Delta}(x) f_{\Lambda_hn}(h x) =  \int_{h x\in h \Sigma_{n,t}} \sp\sp\sp dh \Sigma_{n,t}(h x) 1_{\Delta}(x)  f_{\Lambda_hn}(h x) $$ $$= \int_{y\in h \Sigma_{n,t}} \sp\sp\sp dh\Sigma_{n,t}(y) 1_{\Delta}(h^{-1} y) f_{\Lambda_hn}(y)
= \int_{y\in h \Sigma_{n,t}} \sp\sp\sp dh \Sigma_{n,t}(y) 1_{h \Delta}(y) f_{\Lambda_hn}(y) $$
$$ = \int_{\Sigma_{\Lambda_h n,t_h}} \sp\sp\sp d\Sigma_{\Lambda_h n,t_h}(y) 1_{h\Delta}(y) f_{\Lambda_hn}(y)
=  \int_{\Sigma_{\Lambda_h n,t_h}} \sp\sp\sp d\Sigma_{\Lambda_h n,t_h}(x) 1_{h\Delta}(x) f_{\Lambda_hn}(x)\:.$$
The found identity, used in (\ref{interm}) and taking  (\ref{ACT1}) into account  leads to
$$ \langle \psi'| (U_{h}  \sQ_{n,t}(\Delta) U_{h}^{-1} - \sQ_{\Lambda_h n,t_h}(h\Delta)) \psi \rangle =0\quad \mbox{if $\psi ,\psi'\in {\cal S}({\cal H})$.}$$
Since ${\cal S}({\cal H})$ is dense in ${\cal H}$ and the operators are bounded and everywhere defined, the found identity extends to the general case $\psi ,\psi'\in {\cal H}$ ending the proof. \\
\end{proof}

\begin{definition} The family $\{\sQ_{n,t}(\Delta)\}_{\Delta\in \cB(\Sigma_{n,t})}$ constructed in Proposition \ref{FPROP} is the {\bf Newton-Wigner PVM} of the massive Klein-Gordon particle in the reference frame $n$ at time $t$. The collection $\sQ$  of all these PVMs when $n\in  \sT_+$, $t\in \bR$ is the {\bf Newton-Wigner spatial localization observable}.\\
\end{definition}

\begin{remark}
{\em \begin{itemize}
\item[(1)] In view of the $IO(1,3)_+$ covariance and (\ref{TE})
\beq
\sQ_{n,t}(\Delta+t) = 
 U_t^{(n)\dagger} \sQ_{n,0}(\Delta)  U_t^{(n)}\:,\quad \forall  t\in \bR\:\:, \forall \Delta \in \cB(\Sigma_{n,0})\:.
\eeq
In other words, the Newton-Wigner PVM at time $t$ in $n$ is the {\em Heisenberg evolution} of the one at  time zero according to the time evolutor in the reference frame $n$.
\item[(2)] The non-relativistic limit for a state $\psi \in {\cal H}$, in a reference frame $n\in  \sT_+$,  can be viewed as the requirement that $|\psi(p)|$ vanishes  outside a region where  $|\vec{p}_n|$ is strictly  narrowed around $m$.   It is easy to see from (\ref{wave})  that, in this situation, $m\varphi_\psi$ tends to become a standard Schr\"odinger wavefunction for a free particle of mass $m$. The use of same type of states in  (\ref{interm}), shows that 
$\langle \psi| \sQ_{n,0}(\Delta)\psi\rangle$ tends to the probability of finding the particle in $\Delta$ (at $t=0$)  according to  the standard non-relativistic position PVM  on the said state $\psi$. 
\item[(3)] There is however another regime where the Newton-Wigner PVM approximates the PVM of the classical position observable.
It is when $\psi$ is sharply narrowed around a  value of the momentum $p_0$. In that case, similarly to before, 
 $E(p_0)\varphi_\psi$ tends to become a standard Schr\"odinger wavefunction for a free particle of mass $m$ and 
$\langle \psi| \sQ_{n,0}(\Delta)\psi\rangle$ tends to the probability of finding the particle in $\Delta$ (at $t=0$)  according to  the standard non-relativistic position PVM.  \hfill $\blacksquare$
\end{itemize}} 
\end{remark}

\subsection{NW localization does not mean localized covariant  wavefunctions: Antilocality}
I am in a position  to illustrate an annoying  fact which sharply distinguishes the relativistic and  the non-relativistic theory. Newton-Wigner localization in a bounded set $\Delta\subset \Sigma_{n,t}$  for a state $\psi$ implies  that the associated wavefunction $\varphi_\psi$ is essentially supported also {\em outside} $\Delta$ itself at time $t$. 

Choose  a reference frame $n$ and a co-moving Minkowskian coordinate system $t=x^0,x^1,x^2,x^3$ and  wrote $\vec{x}:= (x^1,x^2,x^3)$. 
Looking at (\ref{CONVNW2}), if $\psi \in {\cal H}$, 
\beq
\Psi_t :=\left( {\cal F}^{-1}S_n U^{-1}_{(I,tn)} \psi\right) \in  L^2(\bR^3,d^3x) \label{unit}
\eeq
Notice that $\Psi_t \in \cS(\bR^3)$ if $\psi \in {\cal S}({\cal H})$ where $\bR^3$ identifies with $\Sigma_{n,t}$.
On account of  (\ref{CONVNW2}),  the action of $\sQ_{n,t}(\Delta)$ on $\Psi$ is trivially the multiplication with $1_{\Delta}(\vec{x})$. On the other hand,  the definition of {\em covariant wavefunction} associated to a state   (\ref{wave}) can be re-formulated  in terms of $\Psi$:
\beq
\varphi_{\psi}(t,\vec{x}) := (\overline{- \Delta + m^2})^{-1/4}\Psi_t(\vec{x})\:, \qquad   \psi \in{\cal H} \label{antiloc}\:.
\eeq
This definition  is valid for $\psi \in {\cal S}({\cal H})$ as the original version  (\ref{wave}) is.  However, as indicated,   it can be trivially extended to the general case $\psi \in {\cal H}$,
since the  selfadjoint operator $(\overline{- \Delta + m^2})^{-1/4}$ is bounded and everywhere defined in $L^2(\bR^3, d^3x)$. In that case, the covariant wavefunction satisfies\footnote{According to Section \ref{L2} the vector $\varphi_\psi(t,\cdot)$ can be viewed  in terms of  a representative given by a Lebesgue measurable or a Borel measurable function and one interprets "a.e." accordingly.}  $\varphi_\psi(t,\cdot) \in L^2(\bR^2,d^3x)$. 
A crucial property known as {\em antilocality}  \cite{anti1,murata} of $(\overline{- \Delta + m^2})^{\alpha}$ plays a fundamental role in the rest of the paper.\\

\begin{theorem}\label{teorem11} Let $k \in \bN$, $m>0$, and  suppose that $\bR \ni \alpha\not \in \bZ$.
If both $\Psi\in L^2(\bR^k, d^kx)$ and $(\overline{- \Delta + m^2})^{\alpha} \Psi$ 
 vanish a.e. with respect to $d^kx$  in an open non-empty set $\Omega \subset \bR^k$ -- assuming $\Psi \in D((\overline{- \Delta + m^2})^{\alpha})$ for $\alpha>0$ -- then  $\Psi =0$ in $ L^2(\bR^k, d^kx)$.\\
\end{theorem} 

\noindent This theorem together with   Eq.(\ref{antiloc}) permit to prove a well-known annoying fact  regarding spatial  localization according to NW: localized states do not correspond to localized covariant wavefunctions (item (2) below).\\

 \begin{proposition}\label{PROP12}  Let us consider the Newton-Wigner localization observable $\sQ$  of a massive Klein-Gordon particle. The following facts are true for given $n\in  \sT_+$, $t\in \bR$.
 \begin{itemize}
 \item[(1)]  $\sQ_{n,t}(\Delta)=0$ if and only if $\Delta$ has zero measure with respect to $d\Sigma_{n,t}$.
\item[(2)] Let  $\psi \in {\cal H}\setminus \{0\}$ be  localized in a  spatial region $\Delta \in \cB(\Sigma_{n,t})$, i.e.,  $$\sQ_{n,t}(\Delta)\psi = \psi\:.$$ 
If $\Delta$ is not dense (in particular if $\Delta$ is bounded)  then    $\varphi_\psi(t, \cdot)$ cannot vanish a.e. in every fixed  non-empty open subset of $\Sigma_{n,t}\setminus \Delta$.
 \end{itemize}
 \end{proposition}

  \begin{proof}
  (1) is obvious since, under unitary equivalence, $\sQ_{n,t}(\Delta)$ is the multiplicative operator  $1_{\Delta}$.
Let us pass to (2).  Suppose that   $\Delta$ is not dense and consider  an open non-empty set $\Omega\subset \Sigma_{n,t}\setminus \Delta$.
$\sQ_{n,t}(\Delta)\psi = \psi$ is equivalent  to $1_{\Delta}\Psi_t = \Psi_t$ a.e. with respect to $d\Sigma_{n,t}$, in particular
 $\Psi_t(\vec{x})=0$ a.e. in $\Omega$. If also  $\varphi_\psi(t, \vec{x})=  (\overline{- \Delta + m^2})^{-1/4}\Psi_t(\vec{x})=0$ a.e. for  $\vec{x} \in \Omega$,
  Theorem \ref{teorem11} applied  to $\Psi= \Psi_t$  for $\alpha =-1/4$ would imply  $ \Psi_t=0$, namely 
 $\psi =0$. This is impossible because $||\psi|| \neq 0$.
\end{proof}
\subsection{The Newton-Wigner position selfadjoint operator}
I pass to define the Newton-Wigner position self-adjoint operators. Given a reference frame $n\in  \sT_+$, choose  a co-moving  Minkowskian coordinate system $t:=x^0,x^1,x^3,x^3$.
Following \cite{NW,Wightman}, I define, the  {\bf Newton-Wigner position selfadjoint operators} in $n$ associated to a co-moving Minkowskian coordinate system with coordinates    $(t:=x^0,x^1,x^3,x^3)$,
\beq N_{n,t}^\alpha := \int_{\Sigma_{n,t}} x^\alpha d\sQ_{n,t}(x)\quad \alpha = 0,1,2,3\:,\label{NWx}\eeq
  where the integration is the standard one according to a PVM (see, e.g. \cite{Moretti2}).\\

\begin{proposition}\label{propX} The  Newton-Wigner position selfadjoint operators (\ref{NWx}) satisfy the following.
\begin{itemize}
\item[(1)]  $\sigma(  N_{n,t}^\alpha) = \sigma_c(  N_{n,t}^\alpha) = \bR$ for every $\alpha = 0,1,2,3$\:.
\item[(2)]  It holds  $D( N_{n,t}^\alpha) \supset {\cal S}({\cal H})$ and more strongly
\beq
  N_{n,t}^\alpha ( {\cal S}({\cal H})) \subset  {\cal S}({\cal H})\:,
\eeq
 and ${\cal S}({\cal H})$ is a core for all those operators.
\item[(3)]  The {\bf Heisenberg commutation relations} hold, where $k,h=1,2,3$:
\beq [ N_{n,t}^k, N_{n,t}^h]|_{{\cal S}({\cal H})} =[  P_{n h}, P_{n k}]|_{{\cal S}({\cal H})} =0\:,  \qquad 
[  N_{n,t}^k, P_{n h}]|_{{\cal S}({\cal H})} = i\delta^k_hI|_{{\cal S}({\cal H})}
\eeq
so that, in particular the statement of the {\bf Heisenberg principle} holds for $h=1,2,3$:
\beq
\Delta_\psi  N^k_{n,t} \Delta_\psi P_{nk} \geq 1/2\:,  \quad \psi \in {\cal S}({\cal H})\:.
\eeq
\item[(4)] The Heisenberg  time evolution relation is valid:
\beq
 U^{(n)\dagger}_t  N_{n,0}^kU^{(n)}_t\psi  =   N_{n,t}^k\psi =   N_{n,0}^k\psi  + t\frac{P_{nk }}{P_{n0}}\psi \quad \mbox{for $\psi \in {\cal S}({\cal H})$ and $k=1,2,3$}\:.\label{EV}
\eeq
\item[(5)] $IO(1,3)_+$ covariance relations are true, if $\psi \in {\cal S}({\cal H})$ and $IO(1,3)_+ \ni h= (\Lambda_h, a_h)$,
\beq
U_h  N_{n,t}^\alpha U_h^{-1} \psi = (\Lambda^{-1}_h)^\alpha_\beta ( N^{\beta}_{\Lambda_h n, t_{h}} - a_h^\beta I)\psi, \quad \forall h \in IO(1,3)_+\:. \label{55}
\eeq
\end{itemize}
\end{proposition}
\begin{proof} See Appendix \ref{APPPROOFS}.
\end{proof}

 If $\psi \in {\cal S}({\cal H})$, property (4) implies that the maps $\bR \ni t \mapsto \langle \psi|  N^{\alpha}_{n,t} \psi\rangle \in \bR^4 \equiv \bM$, $\alpha=0,1,2,3$ is the coordinate description of  
a timelike curve, i.e., the time evolution of  a point in the rest space of $n$ with speed that  is strictly less  than the light speed. In fact, the following corollary holds  which strongly relies on the overall initial  hypothesis $m>0$. That is a sort of {\em Ehrenfest theorem} for the position of a massive free Klein-Gordon particle.\\
 
 \begin{corollary}\label{CORR}
 Let $\psi \in {\cal S}({\cal H})$ satisfy  $||\psi||=1$. The expectation values of the  Newton-Wigner position selfadjoint operators $ N^1_{n,t},  N^2_{n,t},  N^3_{n,t}$ (of a reference frame $n\in  \sT_+$ with  a co-moving Minkowskian coordinate system $t=x^0,x^1,x^2,x^3$) describe a timelike worldline since
 \beq
 \sum_{k=1}^3 \left(\frac{d}{dt} \langle \psi|  N^{k}_{n,t} \psi\rangle\right)^2 < 1 \:.
 \eeq
 \end{corollary}

\begin{proof}
See Appendix \ref{APPPROOFS}.
\end{proof}

I stress that the found result, together with the covariance properties stated in Propositions \ref{FPROP}  and  \ref{propX}  suggest  that the Newton-Wigner position  localization observable  possesses  important physically sound features which should be preserved in any improvement of this sort of formalization. On the other hand, some substantial improvement is also necessary because, as  we shall see  shortly,  the Newton-Wigner   localization also suffers for  physically insurmountable issues related to  causality.

\section{Problems with spatial localization}
This section is devoted to examine   the consequences on the Newton-Wigner position  localization observable  of an important general result by Hegerfeldt \cite{Hegerfeldt, Hegerfeldt2} that, at the end of the play,  rules out it. The analysis only concerns the issue (I1) presented in the introduction and extends to more general notions of spatial localization based on POVMs rather than PVM. 

I stress that I will stick to the basic version of Hegerfeldt's result. A modern formulation,  which improves  original Hegerfeldt's ideas,  appears in \cite{Castrigiano1,Castrigiano2}.

\subsection{Castrigiano's  causality requirement}
Suppose that an one-particle Klein-Gordon pure state represented by $\psi\in {\cal H}$, with $||\psi||=1$, defines a family $\mu^\psi$  of  probability measures $\mu^\psi_{n,t} : \cL(\Sigma_{t,n}) \to [0,1]$ --  where $n\in  \sT_+$, $t\in \bR$  -- such that $\mu^\psi_{n,t}(\Delta)$ represents the probability of detecting the particle in $\Delta \subset \Sigma_{n,t}$. 
I will call this collection a {\bf family of spatial localization probability measures} associated to the state $\psi$. How this association is implemented will be discussed later.

A   physically meaningful requirement on families  of spatial localizations was explicitly introduced by Castrigiano\footnote{Other papers including \cite{Hegerfeldt} and \cite{Castrigiano1, Terno}, use only the requirement (a).} in \cite{Castrigiano2} and therein deeply analyzed in the case of particles with spin (within the more elaborated notion of {\em causal system}).  Castrigiano's  requirement was actually  formulated in terms of POVMs I will introduce later. Here I adopt a definition  in terms of families of probability measures which is equivalent to  Castrigiano's one as soon as one passes to deal with POVMs.  \\
The next definition illustrates Castrigiano's causality requirement corresponding to item (b) in the definition below. The notion of {\em causal time evolution} presented  in (a) was also introduced by Castrigiano. I stress that the  distiction between of (a) and (b)  is just functional to this  study, though the validity of (a) is an evident consequence of  (b)\footnote{I am grateful to  Prof. Castrigiano for clarifications on these issues.} which is the causality condition introduced in \cite{Castrigiano2}.\\

\begin{definition}\label{LCR}
{\em  Let
  $$ \mu^\psi:= \{\mu^\psi_{n,t} : \cL(\Sigma_{t,n}) \to [0,1]\}_{n\in  \sT_+, t\in \bR} $$  be the family of spatial localization probability measures  of a pure state represented by $\psi \in {\cal H}$ with $||\psi||=1$.
 \begin{itemize}
 \item[(a)] A given  $n\in  \sT_+$ defines  a {\bf causal time evolution}
  if, for every 
$\Delta \in \cL(\Sigma_{n,t})$,
\beq
\mu^\psi_{n,t}(\Delta) \leq \mu^\psi_{n,t'}(\Delta') \quad \forall  t'\in \bR \label{loccaus}\:.
\eeq
where  $\Delta' := \left(J^+(\Delta)  \cup J^-(\Delta) \right)\cap \Sigma_{n,t'}$. 
 \item[(b)]  {\bf (Castrigiano's causality requirement)} The full family  $\mu^\psi$ is {\bf causal} if,   for every 
$\Delta \in \cL(\Sigma_{n,t})$, it holds
 \beq
\mu^\psi_{n,t}(\Delta) \leq \mu^\psi_{n',t'}(\Delta') \quad \forall n, n'\in  \sT_+ \:,  \forall  t, t'\in \bR \label{loccaus2}\:,
\eeq
where  $\Delta' := \left(J^+(\Delta)  \cup J^-(\Delta) \right)\cap \Sigma_{n',t'}$. \\
 \end{itemize}}
\end{definition}

\begin{remark}\label{REML}
{\em \begin{itemize}
\item[(1)] The reason why I passed from $\cB(\Sigma_{t,n})$ to $\cL(\Sigma_{t,n})$ is that, if $\Delta \in \cB(\Sigma_{t,n})$ then it may happen that $\Delta'  \not \in \cB(\Sigma_{t',n'})$. {\em Vice versa}, if $\Delta \subset \Sigma_{n,t}$ (non necessarily Lebesgue measurable!), then  $\Delta'   \in \cL(\Sigma_{t',n'})$ for every $n'\neq n$ and $t,t'\in \bR$ as established in Lemma 16 \cite{Castrigiano2}.
\item[(2)]  Evidently, the validity of  (b) implies that (a) holds for every choice of $n\in  \sT_+$. However, if (a) is true for all $n\in  \sT_+$, (b) can be false in principle. 
\item[(3)] The definition of causal family of spatial localizations  is symmetric under time reversal, i.e., it  also consider $J^-(\Delta)$. This is because, if interpreting the probability ans a density of particles, the particles which reached $\Delta$
   at time $t$ must have passed through $J^-(\Delta)\cap \Sigma_{n',t'}$ for every rest space $\Sigma_{n',t'}$ in the past of $\Delta$. There are intermediate situations where the intersection of  $\Sigma_{n',t'}$ and  $\Sigma_{n,t}$ includes $\Delta$ but they can be treated separately by dividing the particles into two cases.  \hfill $\blacksquare$
\end{itemize}}
\end{remark}

\subsection{Justification of the causal condition  in the special case of sharp localization}\label{secjust} The condition (b) above seems  physically reasonable. However it is not obvious how to justify it within the framework  of this work (and the analogous ones), {\em as  everything should be justified within the framework of the issue (I1) disregarding (I2)}.  In other words,  I should not  to refer to any  issue concerning post-measurement states, but I have   to stick to a unique given family $\mu^\psi$.  I can at most perform one position measurement because,  after a measurement, referred to the state $\psi$ and the family $\mu^\psi$,  the state changes\footnote{Referring  to general, quite realistic,  measurement instruments,  the post measurement state is not  pure even if the initial state is.} $\psi \to \psi'$  and the family $\mu^\psi$ changes accordingly  $\mu^\psi \to \mu^{\psi'}$, into a way I cannot control without a precise choice of the post-measurement state. 
Instead, Definition \ref{LCR} considers a unique family $\mu^\psi$.  

There is a case however where a justification of the requirements in the above definition is sufficiently  easy even referring to a unique family $\mu^\psi$ (one measurement procedure only).
 Let us illustrate how the failure of  condition (a) (thus (b)) for a choice of $n\in  \sT_+$ would  permit superluminal transmission of information in the special case where there are states {\em strictly localized} at time $t$  in some bounded regions $\Delta$. In other words,  $\mu^\psi_{n,0}(\Delta)=1$ in the reference frame $n\in  \sT_+$.   This justification does not need to tackle the issue of the  post measurement state. 

Consider two types of Klein Gordon particles with masses $m_1 \neq  m_2$ 
 respectively and collect, at $t=0$, a large number of these particles (of the two types) in a box at rest in $\Sigma_{n,0}$. We can image  the box as the bounded region  $\Delta\subset 
 \Sigma_{n,0}$.  
 I assume that it is possible to open the box only for the mass $m_1$ or mass $m_2$ particles with some sort of filter. Next the procedure is
 
 (1)  I make a decision about  which type of particles ($m_1$ or $m_2$)  to free  from $\Delta$
  at time $t=0$ and I free it;
  
 (2)  somebody  detects the particles in $\Sigma_{n,t}$ at time $t>0$ and observes  the value of the mass.\\
{\em If (\ref{loccaus}) failed, a particle could be detected in the region  $\Delta' \subset  \Sigma_{n,t}$ with  $\Delta' \cap  J^+(\Delta)= \emptyset$, and 
this procedure would manage to  transmit the information about my mass choice made  in the spatial region  $\Delta$ at time $t=0$ outside the  causal future of this event!}
 
The crucial point in the above discussion is that some states are at disposal whose probability measure at $t=0$ is  zero outside the bounded region $\Delta$.   

Very unfortunately, as I will  discuss shortly, sharply localized position probabilities are ruled out by the Hegerfeldt theorem.

 The above  justification of the causality condition (b) which only  relies on (I1)  does not seem to be that easy to re-propose   if referring to families $\mu^\psi$ which are not sharply localized (see Sect. \ref{secALS}). In this case, as discussed in the rest of  this work, the position observable is described in terms of a POVM instead a PVM. 
In principle, in the absence of sharply localized states,   one may try to use again an analogous argument where, at time $t=0$,  two types of bosons stay in a box with a certain {\em very large probability}.  Opening the box for only one kind of boson should be formalized in terms of suitable {\em quantum operations} \cite{Buschbook}, not necessarily trace preserving, which define the quantum states of the two types of particles  at time $t>0$. Here,  precise theoretical choices seem to be necessary and the elementary setting of (I1) does not seem to be sufficient.
 
 This matter deserves further attention, but in this paper I will be content with  assuming  Castrigiano's causality  requirement and the consequent notion  of causal time evolution  as natural ideas.

\subsection{Spatial localization in terms of  POVMs}

As is known,  (see, e.g., \cite{Moretti2}), if $A :{\cal H} \to {\cal H}$, then $A\geq 0$ means $\langle \psi|A\psi \rangle \geq 0$ for all $\psi \in {\cal H}$.
This requirement for $A$ is equivalent to $A=A^\dagger \in \gB({\cal H})$ and $\sigma(A) \subset [0,+\infty)$. Finally, if also $B: {\cal H}\to {\cal H}$, then $A\geq B$ means $A-B \geq 0$.

An {\bf effect} (see \cite{Buschbook} for a modern up-to-date  textbook  on the subject)  is a bounded  operator $\sE\in \gB({\cal H})$, for  a Hilbert space ${\cal H}$, such that $0\leq \sE \leq I$. $\gE({\cal H})$ will indicate henceforth  the set of effects in ${\cal H}$. 
An orthogonal projector is an effect but there are effects which are not orthogonal projectors. 

A (normalized) {\bf Positive Operator Valued Meaure}  (POVM) is a map $$\Sigma(X) \ni \Delta \mapsto \sE(\Delta) \in \gE({\cal H})\:,$$  where $\Sigma(X)$ is a $\sigma$-algebra on $X$, such that the function  is  (see Def. 4.5 in  \cite{Buschbook} and the remarks under that definition)
\begin{itemize}
\item[(a)]
normalized: $\sE(X) =I$;
\item[(b)] $\sigma$-additive:  $\sum_{n \in \bN} \sE(\Delta_n) = \sE(\cup_{n\in \bN}\Delta_N)$
when $\Delta_n\cap \Delta_m = \emptyset$ for $n\neq m$ and the sum is understood in the weak (or equivalentely strong) operator topology.  
\end{itemize}
Notice that (a) and (b) imply in particular that $\sE(\emptyset)=0$. Furthermore (b) can be equivalently replaced by the requirement  that $\Sigma(X) \ni \Delta \mapsto \langle \psi| \sE(\Delta) \psi' \rangle$ is a complex measure (with finite total variation) for every $\psi,\psi' \in {\cal H}$.\\

\begin{remark}
{\em \begin{itemize}
\item[(1)] It is clear that a PVM is a specific case of POVM where the positive operators $\sE(\Delta)$ are orthogonal projectors. 
\item[(2)]  A POVM  does not satisfy in general  $[\sE(\Delta),\sE(\Delta')]=0$ for $\Delta \cap \Delta' = \emptyset$ contrarily to what happens for a PVM. 
\item[(3)] The one-to-one link between selfadjoint operators and PVMs does not hold in case of POVMs. Something remains however, since under some technical hypotheses a POVM is uniquely  determined  by a symmetric operator, in terms of  the {\em first moment} of the POVM, as I will briefly discuss later. This fact, the failure of the hypotheses for that property, will play  some role in this paper.  \hfill $\blacksquare$\\
\end{itemize}}
\end{remark}

The general notion of observable, in the modern approaches to Quantum Theory,  is a (normalized) POVM on a $\sigma$-algebra $\Sigma(X)$ and taking values in $\gB({\cal H})$, where ${\cal H}$ is the Hilbert space of the considered quantum system:
\begin{itemize}
\item[(1)] The elements $\Delta \in \Sigma(X)$ are the outcomes of measurements and,
\item[(2)]  if $\rho$ is a generally mixed state -- a trace class, unit-trace positive operator in $\gB({\cal H})$; $\Sigma(X) \ni \Delta \mapsto tr(\rho \sE(\Delta))$ is the probability measure associated to these outcomes. It boils down to $\Sigma(X) \ni \Delta \mapsto \langle \psi|\sE(\Delta) \psi\rangle$ in case of a pure state represented by the unit vector $\psi \in {\cal H}$. 
\end{itemize}

\begin{definition}\label{POVMloc}
A {\bf relativistic spatial localization observable}   for a  Klein-Gordon particle of mass $m>0$ described in the (complex, separable) Hilbert space ${\cal H}$ is defined as a family of normalized  POVMs $\sE_{n,t}: \cL(\Sigma_{n,t}) \to \gE({\cal H})$ , where  $n\in  \sT_+$ and $t\in \bR$,
that  is {\bf covariant} with respect to the  strongly continuous unitary representation $U$ of $IO(1,3)_+$ (\ref{ACT1}):
 \beq
  U_{h}  \sE_{n,t}(\Delta) U_{h}^{-1} = \sE_{\Lambda_h n,t_h}(h\Delta) \:, \quad \forall \Delta \in \cL(\Sigma_{n,t})\:, \quad h \in IO(1,3)_+\label{ACOVNW}\:.
\eeq
\end{definition}

A very detailed technical analysis of the notion above (called {\em Poincar\'e covariant POL} therein) appears in sections 6 and 7 of \cite{Castrigiano2} referring to a general system and establishing some extension and uniqueness properties from POVMs covariant under the Euclidean group to POVMs covariant under the full $IO(1,3)_+$ group.

The use of POVMs defined on $\cL(\Sigma_{n,t})$ is mandatory due to  Remark \ref{REML}. 

With the same elementary  procedure to complete positive measures, a POVM $\sE$ defined on 
$\cB(\Sigma_{n,t})$ uniquely extends to a {\bf completion}: another POVM $\overline{\sE}$, on a larger $\sigma$-algebra  $\overline{\cB(\Sigma_{n,t})}^{\sE}$ made of the unions of the elements of $\cB(\Sigma_{n,t})$ with the subsets of the zero-$\sE$-measure sets, 
$$\overline{\sE}(\Delta \cup Z) :=  \sE(\Delta)\:,  \quad \Delta \in \cB(\Sigma_{n,t})\:, \quad Z \subset B \in  \cB(\Sigma_{n,t})\:, \quad  \sE(B)=0\:.$$
Exactly as in standard measure theory, $\overline{\cB(\Sigma_{n,t})}^{\sE}$ is characterized  by the fact that it is the smallest $\sigma$-algebra including  $\cB(\Sigma_{n,t})$ and equipped with an extension $\overline{\sE}$ of $\sE$ such that all subsets  of zero-$\overline{\sE}$-measure sets in  $\overline{\cB(\Sigma_{n,t})}^{\sE}$ belong to  $\overline{\cB(\Sigma_{n,t})}^{\sE}$.

Trivially, the outlined procedure extends a POVM which is a PVM  to a completion that  is a  PVM as well.  In particular, the  completion of the previously discussed  Newton-Wigner PVM
turns out to be defined on $\overline{\cB(\Sigma_{n,t})}^{\sQ_{n,t}}= \cL(\Sigma_{n,t})$ as a consequence of (1) Proposition \ref{PROP12} and elementary properties of the  Lebesgue measure:
 $\cL(\Sigma_{n,t})\ni \Delta \mapsto \overline{\sQ_{n,t}}(\Delta) \in \gE({\cal H})$.
 This completion still satisfies the $IO(1,3)_+$ covariance and all the properties established in the previous section as one immediately proves.  In the rest of the paper I will simply write $\sQ_{n,t}(\Delta)$ in place of $\overline{\sQ_{n,t}}(\Delta)$ when $\Delta \in \cL(\Sigma_{n,t})$.

\subsection{Troubles with Newton-Wigner and sharply localized states: the Hegerfeldt theorem}
Hegerfeldt \cite{Hegerfeldt2} proved the following quite devastating  theorem against the Newton-Wigner notion of localization in particular. I reformulate the result established in  \cite{Hegerfeldt2} into the language of Definition \ref{LCR} and explicitly for a massive Klein-Gordon real spinless particle.\\

\begin{theorem}[Hegerfeldt]\label{HGT} Consider a spatial localization POVM  of a massive Klein-Gordon particle according to Def. \ref{POVMloc}.
Suppose that  there are   $\psi  \in {\cal H}$ with $||\psi||=1$ and  $e\in \Sigma_{n_e,t_e}$ such that the probability to find the particle outside the  balls $B_r(e) \subset \Sigma_{n_e,t_e}$ with common center  $e$ and variable radii $r>0$ satisfies 
 $$\langle \psi| \sE_{n_e,t_e}(\Sigma_{n_e,t_e} \setminus B_r(e)) \psi \rangle \leq K_1 e^{-K_2 r} \quad \mbox{for some $K_1>0$, $K_2 \geq 2m$ and all $r>0$}\:.$$
 Then $n_e$ cannot define a causal time evolution of
 the family of probability measures $\mu^{\psi}_{n,t}:= \langle \psi| \sE_{n,t}(\Delta) \psi \rangle$
  according to condition (a) in  Def. \ref{LCR}.\\
\end{theorem}
%

\noindent A crucial corollary follows against the Newton-Wigner notion of spatial localization.\\

\begin{corollary}\label{cor20} The (completion of the) Newton-Wigner spatial localization observable does not satisfy Castrigiano's  causality condition, because (a) in  Def. \ref{LCR} fails for every choice of $n\in  \sT_+$\:.
\end{corollary}

\begin{proof}  Arbitrarily fix $e\in \Sigma_{n_e,t_e}$,  choose $R>0$ and consider the orthogonal projector $\sQ_{n_e,t_e}(B_R(e))$. It holds $\sQ_{n_e,t_e}(B_R(e)) \neq 0$
 due to (1) in Proposition \ref{PROP12}, since an open has strictly positive measure $d\Sigma_{n_e,t_e}$. Therefore there exists $\psi = \sQ_{n_e,t_e}(B_R(e))\psi$ with $||\psi||=1$. Evidently $\langle \psi| \sQ_{n_e,t_e}(\Sigma_{n_e,t_e} \setminus B_r(e)) \psi \rangle=0$ if $r>R$ since $$ \sQ_{n_e,t_e}(\Sigma_{n_e,t_e} \setminus B_r(e)) \psi=
 \sQ_{n_e,t_e}(\Sigma_{n_e,t_e} \setminus B_r(e)) \sQ_{n_e,t_e}(B_R(e))  \psi  = \sQ_{n_e,t_e}(\emptyset)\psi=0\:.$$
$\psi$ satisfies the hypotheses of Hegerfeldt's theorem with respect the family of balls $B_r(e)$.  Arbitrariness of $n_e \in  \sT_+$ concludes the proof.
\end{proof}

An interesting paper by Ruijsenaars \cite{Ruijsenaars} presents some explicit numerical  estimates of the probabilities of recording  a violation of causality through measurements of the Newton-Wigner observable for a scalar  Klein-Gordon massive particle.

It is evident that,  on account of the corollary, {\em Physics rules out the Newton-Wigner notion of localization} because it does not satisfy a  basic requirement about causality,  {\em in particular  taking Sect. \ref{secjust} into account}.
 However,  this is very  disappointing because {\em the Newton-Wigner  position operator  shows some natural and quite appealing features},  as  previously illustrated in Proposition \ref{propX} and its Corollary \ref{CORR}. This inconclusive asymmetry is very annoying and is certainly a  reason why  Newton Wigner's notion of localization is still a subject of discussion in the literature. In the rest of the paper will see how it is possible to keep the good things (the position operator)  and get rid of the bad ones (the PVM). \\

\begin{remark}{\em 
\begin{itemize}
\item[(1)] There are other, even more severe,  problems with the Newton-Wigner notion of spatial localization and causality   when one analyses  it on the ground of the issue (I2) of the introduction,  by assuming the L\"uders' projection postulate about the post-measurement state.

\item[(2)] The Newton-Wigner notion of spatial localization is acausal not only with respect to time translations but equally
regarding boosts. This so-called frame-dependence of
Newton-Wigner localization  has been observed already in \cite{WM}  and it is still studied in the
literature, e.g. \cite{FKW}. It is obvious that any
notion of spatial localization in terms of POVMs should meet the requirement of frame-independence.

\item[(3)] It is interesting to notice that the example of the rejection of the Newton-Wigner observable shows  how the idea that every PVM/selfadjoint operator in the  Hilbert space of a quantum system  must be  be an observable  is definitely untenable. However to author's knowledge this is the first time that, in quantum mechanics,  the rejection of a selfadjoint operator as an observable in quantum mechanics  is due to  local-causality and not to the existence of a gauge group or a superselection rule. 

\item[(4)] The above version of the Hegerfeldt theorem is the classic one, it  explicitly refers to the Klein-Gordon particle and  can be immediately extended to particles with spin. Actually it is not necessary that full covariance with respect to our  representation of $IO(1,3)_+$ holds. There are more abstract versions of this theorem that refer to abstract POVMs and  rely only on (a)  positivity of the selfadjoint generator of temporal translations and (b) covariance with respect to four translations.   
See in particular Theorem B1\footnote{That theorem includes the  hypothesis ``{\em $\langle \psi| \sE_\Delta \psi \rangle=1$ and $\langle \varphi|\sE_\Delta \varphi \rangle =0$ implies $\langle\psi|\varphi\rangle =0$ }''. However it  is not necessary since it is automatically satisfied by every POVM $\sE$.} in \cite{tesi}. A throughout analysis of the interplay of spatial localization and Hamiltonian positivity appears in sections 4 and 5 of \cite{Castrigiano2}.  \hfill $\blacksquare$
\end{itemize}}
\end{remark}

\section{The spatial localization observable  proposed by Terno}
In \cite{Castrigiano2}, Castrigiano proved that for spin $1/2$ it is possible to define a spatial localization  observable different from the Newton-Wigner one which satisfies the causality requirement (b) of Def. \ref{LCR}. That observable is a PVM if the  positivity assumption on the Hamiltonian evolutor is not imposed and becomes a POVM when restricting to the subspace of positive energy. Unfortunately, that  construction does not work for scalar Klein-Gordon particles  as discussed in Section 23 of  \cite{Castrigiano2}.

\subsection{Terno's POVM: the  heuristic definition from QFT}\label{secINTROT}
Terno \cite{Terno} introduced  a position localization  POVM starting from elementary notions of  free QFT in Minkowski spacetime.
Though that notion was also extended to photons in \cite{Terno}, here I stick to the case of a real scalar  massive Klein-Gordon field.
 
I review the definition of that POVM in the formal language of theoretical physics of QFT  first. Later I will  translate it into a more mathematically rigorous setting. I start from the stress energy operator of QFT. Let
$$:\hat{T}_{\mu\nu}\spa:(x) :=\: :\spa\partial_\mu \hat{\phi} \partial_\nu \hat{\phi}\spa:(x) - \frac{1}{2} g_{\mu\nu} \left( :\spa\partial_\alpha \hat{\phi}\partial ^\alpha \hat{\phi}\spa:(x) + m^2 :\spa\hat{\phi}^2\spa:(x)\right)$$ be the coordinate-representation of the normally ordered {\em stress-energy tensor operator} in the symmetric Fock space  $\gF_+({\cal H})$ of the real {\em Klein-Gordon field operator} $\hat{\phi}$ with mass $m>0$ . Referring to a Minkowski coordinate system co-moving with $n\in  \sT_+$, if   $\Delta \subset \Sigma_{n,t}$, define
\beq\label{Ternoinformal}
\sA_{n,t}(\Delta):= \frac{1}{\sqrt{H_n}}P_1 \int_\Delta :\spa\hat{T}_{\mu\nu}\spa:\spa(x) n^\mu n^\nu \:  d\Sigma_{n,t}(x) P_1\frac{1}{\sqrt{H_n}} \:, \quad \mbox{with $-n\cdot x = t$},
\eeq
where $P_1 : \gF_+({\cal H}) \to {\cal H}$ is the orthogonal projector onto the one-particle space of the symmetric Fock space $\gF_+({\cal H})$ constructed upon the Minkowski vacuum state with ${\cal H}$ as the one-particle subspace. Actually the definition in \cite{Terno} uses the total Hamiltonian in the Fock space and $P_1$ is swapped with the inverse square root of the said Hamiltonian, but that definition is formally equivalent to that above.

Formally speaking, without paying attention to domains, as $:\spa \hat{T}_{\mu\nu}\spa: (x) n^\mu n^\nu$ turns out to be  positive,  the integral is a positive operator so that $0\leq \sA_{n,t}(\Delta) \leq \sA_{n,t}(\Delta')$ if $\Delta \subset \Delta'$.
The integral  on the whole rest space amounts to $$\sA_{n,t}(\Sigma_{n,t}) =H_n^{-1/2} P_1  \left[0 \oplus H_n \oplus (H_n\otimes I \oplus I \otimes H_n) \oplus \cdots \right]   P_1 H^{-1/2}_{n} = H_n^{-1/2} H_n H^{-1/2}_{n}=I\:.$$ Hence $0\leq \sE(\Delta)\leq I$.  $\sigma$-additivity with respect to   $\Delta$ is guaranteed by the very presence of the integration over $\Delta$. As a matter of fact, barring mathematical details I will fix shortly, that is a  (non-commutative) POVM. 

A straightforward formal manipulation of the right-hand side of (\ref{Ternoinformal}), yields also a natural $IO(1,3)_+$ -covariance relation
\beq\nonumber 
U_h \sA_{n,t}(\Delta) U^\dagger_h = \sA_{\Lambda_h n, t_h}(h\Delta) \quad \mbox{if $h\in IO(1,n)_+$} \:.
\eeq

The physical idea behind Terno's definition should be  clear:  {\em probabilistically speaking,  the particle stays where the energy is}.
This idea was previously formulated in \cite{BK}, where however no explicit POVM was constructed.  The crucial normalization factors $H_n^{-1/2}$ were explicitly introduced in \cite{Terno}.

\subsection{Terno's spatial localization observable }
Expanding the quantum field in modes as usual, a straightforward computation starting from (\ref{Ternoinformal}) yields, for $\Delta \in \cL(\Sigma_{n,t})$ and $\psi \in {\cal S}({\cal H})$, 
$$\sA_{n,t}(\Delta)\psi  :=$$
\beq \label{TPOVM}
\int_\Delta \sp d\Sigma_{n,t}(x) \spa \int_{\sV_{\mu,+}} \sp\sp\sp d\mu_m(p) \frac{e^{-i(q-p)\cdot x}}{(2\pi)^3} \frac{\left(E_n(p)E_n(q)+ \frac{1}{2}(p\cdot q+ m^2)\right)}{\sqrt{E_n(p)E_n(q)}} \psi(p)\quad \mbox{with $-n\cdot x = t$},
\eeq
which I will  assume to be the  definition of the family of operators $\sA_{n,t}(\Delta)$, for $n\in \sT_t$ and $t\in \bR$,  on the domain ${\cal S}({\cal H})$.\\

\begin{theorem} \label{TEO0} Referring to a massive real  Klein-Gordon particle,  the family of operators $\sA_{n,t}(\Delta) : {\cal S}({\cal H}) \to {\cal H}$ defined in   (\ref{TPOVM}) for $n\in  \sT_+$, $t\in \bR$, $\Delta \in \cL(\Sigma_{n,t})$ uniquely  continuously  extends to a POVM -- we shall indicate with the same symbol -- for every given pair $n,t$. The following further facts are valid.
\begin{itemize}
\item[(1)] The family is  covariant with respect to the  strongly continuous unitary representation $U$ of $IO(1,3)_+$ (\ref{ACT1}):
 \beq
  U_{h}  \sA_{n,t}(\Delta) U_{h}^{-1} = \sA_{\Lambda_h n,t_h}(h\Delta) \:, \quad \forall \Delta \in \cL(\Sigma_{n,t})\:, \quad \forall h \in IO(1,3)_+\label{ACOVNWT}\:.
\eeq
and thus it defines a relativistic spatial localization observable.
\item[(2)]  Referring to the (Lebesgue-completion of the) Newton-Wigner spatial localization observable $\sQ_{n,t}$, the following identity is true
\beq \label{TPOVM2}
\sA_{t,n}(\Delta) = \sQ_{t,n}(\Delta) + \frac{1}{2}\left( \eta^{\mu\nu}\frac{P_{n\mu}}{H_n}  \sQ_{n,t}(\Delta) \frac{P_{n\nu}}{H_n}  + \frac{m}{H_n} \sQ_{n,t}(\Delta)  \frac{m}{H_n} \right)
\eeq
 for every $n\in  \sT_+$, $t\in \bR$, and $\Delta \in \cL(\Sigma_{n,t})$.  (The various everywhere-defined bounded  composite operators $P_n^\mu/H_{n}$  and $m/H_n$  are defined in terms of the joint spectral measure of $P^\mu$ and with standard spectral calculus.)
\end{itemize}
\end{theorem}

\begin{proof} Let us prove (1) and (2). Fix $n\in  \sT_+$ and $t\in \bR$.
 If $\psi',\psi \in {\cal S}({\cal H})$ and we indicate by $B\psi$ the right-hand side of  (\ref{TPOVM}) and by $C$ the right-had side of (\ref{TPOVM2}), a straightforward computation that takes (\ref{interm}) into account  proves that 
 $\langle \psi'| B\psi\rangle= \langle \psi'|C \psi \rangle$. Since $\psi'$ varies in a dense set, the found   identity implies that $B\psi= C\psi$ for all  $\psi \in {\cal S}({\cal H})$. As $C$ is continuous and everywhere defined on ${\cal H}$, we conclude that the operator defined in   (\ref{TPOVM}) uniquely  extends by continuity to the operator in  (\ref{TPOVM2}). On the other hand, since the operators $\sQ_{n,t}(\Delta)$ define a PVM, the structure of the right-hand side of  (\ref{TPOVM2}), which can be re-arranged to
\beq\sA_{t,n}(\Delta) =  \frac{1}{2}\left(\sQ_{t,n}(\Delta) +\sum_{k=1}^3\frac{P_{nk}}{H_n}  \sQ_{n,t}(\Delta) \frac{P_{nk}}{H_n}  + \frac{m}{H_n} \sQ_{n,t}(\Delta)  \frac{m}{H_n} \right)\label{TPOVM3}\eeq
 defines a family of positive operators of $\gB({\cal H})$. Notice in particular  that  $ \frac{P_{n\nu}}{H_n}= \left(\frac{P_{n\nu}}{H_n}\right)^\dagger \in \gB({\cal H})$ and $ \frac{m}{H_n}= \left(\frac{m}{H_n}\right)^\dagger \in \gB({\cal H})$. The family of operators in the right-hand side of (\ref{TPOVM3}),   is also  evidently weakly $\sigma$-additive in $\Delta \in \cL(\Sigma_{n,t})$. The constructed POVM is normalized because $\sQ_{n,t}$ is:
 $$\sA_{t,n}(\Sigma_{n,t}) = \sQ_{t,n}(\Sigma_{n,t}) + \frac{1}{2}\left( \eta^{\mu\nu}\frac{P_{n\mu}}{H_n}I \frac{P_{n\nu}}{H_n}  + \frac{m}{H_n} I  \frac{m}{H_n} \right) =I+ 0 =I\:.$$ 
 The proof of (\ref{ACOVNWT}) is strictly analogous to the one of (\ref{COVNW}) or it can be established immediately from it by taking (\ref{TPOVM2}) into account and the obvious covariance properties of the operators $P_{n\mu}$.
\end{proof}

\begin{definition} Referring to Theorem \ref{TEO0},
we call each $\sA_{n,t}$  {\bf Terno's spatial localization POVM} in the reference frame $n\in  \sT_+$ at time $t\in \bR$.
The family $\sA$ of POVMs $\sA_{n,t}$ will be named {\bf Terno's spatial localization observable}.\\
\end{definition}

\begin{remark} {\em 
Contrarily to the case of the Newton-Wigner localization, covariance with respect to the spatial Euclidean subgroup is not sufficient to fix the structure of $ \sA_{n,t}$, since there are infinitely many POVMs with that covariance property with respect to a unitary strongly continuous representation of the Euclidean group \cite{POVMEUC}. \hfill $\blacksquare$} 
\end{remark}

\subsection{Almost localized states}\label{secALS}
The following proposition illustrates a fundamental difference between  the notion of spatial localization by Newton-Wigner and the one by Terno:  localized states in bounded regions are permitted by the former but are impossible for  the latter. 
This implies in particular  that the argument of Corollary \ref{cor20} -- which ruled out the Newton-Wigner localization notion --  cannot be directly applied to $\sA_{n,t}$.  In \cite{Terno}, it is proved (exploiting  an argument of \cite{BK}) that the spatial decay of the probability distribution arising from  
 the POVM $\sA_{n,t}$ does not reach the bound sufficient to trigger  Hegerfeld's local-causality catastrophe. I will achieve  that result indirectly, by establishing that the time evolution with respect to every $n\in \sT_+$ is causal for the  said POVM . 

However, it is not the whole story. Indeed, the second statement of the next proposition  shows that, for every (in particular {\em bounded}) region $\Delta \in \cL(\Sigma_{n,t})$ with non-empty interior, there are states which are arbitrary good approximations of states sharply localized in that region.\\

\begin{proposition}\label{PROPNL} 
Referring to the Terno spatial localization observable $\sA$, the following facts are true.
\begin{itemize}
\item[(1)]
Suppose that $\psi \in {\cal H}$ with $||\psi||=1$,   $n\in  \sT_+$,   $t\in \bR$, and  $\Delta \in \cL(\Sigma_{n,t})$ satisfy
$$\langle \psi|\sA_{n,t}(\Delta) \psi\rangle =1\:.$$
In that case
$\Delta$ is dense in $\Sigma_{n,t}$.
 In particular, $\Delta$ cannot be bounded.
\item[(2)] For every given $n\in \sT_+$, $t\in \bR$ and $\Delta \in \cL(\Sigma_{n,t})$ with $Int(\Delta)\neq \emptyset$, there is  a sequence of vectors
$\{\psi_j\}_{j \in \bN}\subset {\cal H}$ such that $||\psi_j||=1$ and $$\langle \psi_j| \sA_{n,t}(\Delta)\psi_j\rangle \to 1\:, \quad \mbox{ as $j\to +\infty$.}$$

\item[(3)] For every given $n\in \sT_t$, $t\in \bR$ and $\Delta \in \cL(\Sigma_{n,t})$, if $Int(\Delta)\neq \emptyset$, then $||\sA_{n,t}(\Delta)||=1$. \\
\end{itemize}
\end{proposition}

\begin{proof} (1) Define $\Delta' :=\Sigma_{n,t}\setminus \Delta$. By additivity,  $\langle \psi|\sA_{n,t}(\Delta') \psi\rangle =0$. From (\ref{TPOVM3}) and  the fact that $\sQ_{nt}$  is a PVM, 
 $\langle \psi|\sA_{n,t}(\Delta') \psi\rangle =0$ can be rephrased to
 $$\frac{1}{2}||\sQ_{n,t}(\Delta')\psi||^2 + \frac{1}{2} \sum_{k=1}^3|| \sQ_{n,t}(\Delta')H_n^{-1}P_{nk}\psi ||^2 + \frac{m^2}{2} ||\sQ_{n,t}(\Delta')H_n^{-1}\psi||^2 =0 \:.$$
 In particular $\sQ_{n,t}(\Delta')\psi=0$ and $\sQ_{n,t}(\Delta')H_n^{-1}\psi=0$. Using the representation (\ref{unit}) of the Hilbert space vectors, these requirements can be restated to
 $1_{\Delta'}(\vec{x}) \Psi_t(\vec{x})=0$ and  $1_{\Delta'}(\vec{x}) (\overline{-\Delta + m^2I})^{-1/2}\Psi_t(\vec{x})=0$. Hence $\Psi_t(\vec{x})=0$ and $ (-\Delta + m^2I)^{-1/2}\Psi_t(\vec{x})=0$ a.e. on $\Delta'$. If $\Delta'$ includes an open non-empty set,  Theorem \ref{teorem11} would imply that $\Psi_t=0$ which is not permitted by hypothesis.\\
(2)  It is evidently sufficient to prove it for the special case  $\Delta = B_R\subset \Sigma_{n,t}$ given by an open  ball of finite radius $R>0$. Indeed, if $\Delta$ admits non-empty interior, then  $\Delta \supset B_R$ for some such ball and thus   $0\leq \langle \psi| \sA_{n,t}(B_R) \psi \rangle \leq \langle \psi| \sA_{n,t}(\Delta) \psi \rangle \leq 1$ if  $||\psi||=1$.
 A sequence of localizing states $\psi_j$ for $B_R$ is also a sequence of localizing states for $\Delta$.
 Finally, we can always assume $t=0$ without lack of generality as the reader can immediately prove using a trivial time translation and exploiting the covariance properties of $\sA$.  So we prove the thesis for the ball $B_R$.  Consider a $C^\infty$ function $\chi\geq 0$ on $\Sigma_{n,0}$ with $supp(\chi) \subset B_R$. Let us  identify $\Sigma_{n,0}$ with $\bR^3$ with a co-moving Minkowski coordinate system of $n$ whose spatial origin is the center of $B_R$.  If  $\vec{a} \in \bR^3$ is a fixed non-vanishing vector and $j\in \bN$,
$$\hat{\chi}_j(\vec{k}) := \frac{1}{(2\pi)^{3/2}}\int_{\bR^3}  e^{-i k\cdot x }  e^{i j \vec{a} \cdot \vec{x} } \chi(\vec{x})\: d^3x \in \cS(\bR^3)\:.$$
Notice that the $L^2$ norm of these vectors does not depend on $j$ and is $||\chi||_{L^2(\bR^3, d^3x)}$.
We can always choose $\chi$ in order that  $||\hat{\chi}_j||_{L^2(\bR^3, d^3k)}=1 = ||\chi||_{L^2(\bR^3, d^3x)}$ for all $j\in \bN$. Finally,  define the family of the unit vectors $\psi_j \in {\cal H}$,
$$\psi_j(k) :=\sqrt{E_n(\vec{k})} \hat{\chi}_j(\vec{k}) \:,\quad j\in \bN\:.$$
From (\ref{firstPVM}),
$$\langle \psi_j | \sQ_{n,0}(B_R) \psi_j \rangle = \int_{B_R} \overline{\chi(\vec{x})} \chi({\vec x})d^3x = ||\chi||^2_{L^2(\bR^3, d^3x)}=1\:.$$
decomposing $\langle \psi| \sA_{n,t}(\Delta) \psi \rangle$ as in  (\ref{TPOVM2}), we have that
$\langle \psi| \sA_{n,0}(\Delta) \psi \rangle - \langle \psi_j | \sQ_{n,0}(B_R) \psi_j \rangle \to 0$ because
\beq  \left\langle \psi_j \left|\left( \eta^{\mu\nu}\frac{P_{n\mu}}{H_n}  \sQ_{n,0}(B_R) \frac{P_{n\nu}}{H_n}  + \frac{m}{H_n} \sQ_{n,0}(B_R)  \frac{m}{H_n} \right) \right.\psi_j \right\rangle \to 0 \quad \mbox{if $j\to +\infty$}\label{limit}\:.\eeq
The proof of the limit above is postponed to Appendix \ref{APPPROOFS}. This concludes the proof of (2), because 
$\langle \psi_j | \sQ_{n,0}(B_R) \psi_j \rangle=1$ as said above.\\
(3)   is  an easy  consequence of (2),  $0\leq \sA_{n,t}(\Delta)= \sA_{n,t}(\Delta)^\dagger \leq I$ and $||\sA_{n,t}(\Delta)|| = 
\sup\{|\langle \psi| \sA_{n,t}(\Delta) \psi \rangle|\:|\: ||\psi||=1\}$.
\end{proof}

\subsection{Interplay of the first-moment operator of $\sA$  and  the NW position operator}
I can now pass  to introduce the {\em first moment} of Terno's POVM, a symmetric operator. I will  prove  in particular that  its closure  coincides with the Newton-Wigner position operator, so that it preserves all the good properties of the Newton-Wigner position operator. \\

\begin{theorem}\label{TEO1}
 Take $n\in  \sT_+$, $t\in \bR$, choose a co-moving  Minkowski  coordinate system $x^0=t,x^1,x^2,x^3$. There is only one   operator $X^\mu_{n,t} : {\cal S}({\cal H}) \to {\cal H}$, for every $\mu:= 0,1,2,3$, completely defined as the first  moment  of the  POVM $\sA_{t,n}$:
\beq\label{chimoment}
\langle \psi| X^\mu_{n,t} \psi \rangle := \int_{\Sigma_{n,t}} x^\mu d\langle \psi| \sA_{n,t}(x) \psi\rangle \:,   \quad \forall \psi \in {\cal S}({\cal H})\quad \mbox{and where $-n\cdot x=t$}\:.
\eeq
The following facts are true.
\begin{itemize}
\item[(1)] $ X^\mu_{n,t}$ satisfies 
\beq\label{rest0}
\langle \psi| X^\mu_{n,t} \psi \rangle = \langle \psi|  N^\mu_{n,t} \psi \rangle\quad \forall \psi \in {\cal S}({\cal H})\:,
\eeq
where $ N^\mu_{n,t}$ is the Newton-Wigner position operator, so that  the further following facts are  valid.
\begin{itemize}
\item[(a)] The identity holds
\beq\label{rest}
 X^\mu_{n,t}=   N^\mu_{n,t}|_{{\cal S}({\cal H})}\:.
\eeq

\item[(b)]  $X^\mu_{n,t}$ is symmetric, essentially selfadjoint and its unique selfadjoint extension is $ N^k_{n,t}$ itself.

\item[(c)]  The  Heisenberg commutation relations  hold, where $k,h=1,2,3$:
\beq\label{CCR}
[ X_{n,t}^k, X_{n,t}^h]|_{{\cal S}({\cal H})} =[  P_{n h}, P_{n k}]|_{{\cal S}({\cal H})} =0\:,  \qquad 
[ X_{n,t}^k, P_{n h}]|_{{\cal S}({\cal H})} = i\delta^k_hI|_{{\cal S}({\cal H})}\:.
\eeq

\item[(d)]  The $IO(1,3)_+$ covariance relations are true, if $\psi \in {\cal S}({\cal H})$ and $IO(1,3)_+ \ni h= (\Lambda_h, a_h)$,
\beq
U_h X_{n,t}^\alpha U_h^{-1} \psi = (\Lambda^{-1}_h)^\alpha_\beta (X^{\beta}_{\Lambda_h n, t_{h}} - a_h^\beta I)\psi, \quad \forall h \in IO(1,3)_+\:. \label{5}
\eeq

\item[(d)] The Heisenberg  time evolution relation is valid\footnote{A similar equation appears as Eq. (A18) in Terno's paper \cite{Terno}.}:
\beq
 U^{(n)\dagger}_t X_{n,0}^k U^{(n)}_t\psi  = X_{n,t}^k\psi =  X_{n,0}^k \psi  + t\frac{P_{nk }}{P_{n0}}\psi \quad \mbox{for $\psi \in {\cal S}({\cal H})$ and $k=1,2,3$}\:.\label{EV2}
\eeq

\item[(e)]  If $\psi \in {\cal S}({\cal H})$ and $||\psi||=1$,  the first-moment operators  define a timelike worldline because\beq
 \sum_{k=1}^3 \left(\frac{d}{dt} \langle \psi| X^{k}_{n,t} \psi\rangle\right)^2 < 1 \:.\label{EV3}
 \eeq
 \end{itemize}
\item[(2)] If $\psi \in {\cal S}({\cal H})$ with $||\psi||=1$ and $k=1,2,3$,
\beq  \int_{\Sigma_{n,t}} (x^k)^2 d\langle \psi| \sA_{n,t}(x) \psi \rangle = \langle \psi| ( N^k_{n,t})^2\psi\rangle + \left\langle \psi\left| \frac{(P_{n0})^2-(P_{nk})^2}{2(P_{n0})^4}\psi\right. \right\rangle  \:.\label{Xsquare}\eeq
As a consequence,
 a corrected version of 
the  Heisenberg inequality  holds for $k=1,2,3$ (restoring the physical constants):
\beq\label{HPN}
\Delta_\psi X^k_{n,t} \Delta_\psi P_{nk} \geq  \frac{\hbar}{2} \sqrt{1 + 2\Delta_\psi P_{n,k}^2 
\left\langle \psi \left|\frac{(P_{n0})^2-(P_{nk})^2}{(P_{n0})^4}\psi\right. \right\rangle}\:,  \quad \psi \in {\cal S}({\cal H})\:. 
\eeq
where $\Delta_\psi X^k_{n,t}$ is the standard deviation of the probability measure  $\cL(\Sigma_{n,t}) \ni \Delta \mapsto \langle \psi|\sA_{n,t}(\Delta)\psi \rangle \in [0,1]$ .
\end{itemize}
\end{theorem}

\begin{proof} 
It is clear that, if an operator $X_{n,t}^\mu$ exists that satisfies (\ref{chimoment}), then  it must be unique on its domain ${\cal S}({\cal H})$. That is because, by polarization any other 
operator $S: {\cal S}({\cal H}) \to {\cal H}$ that satisfies  that identity would have the same matrix elements
$\langle \psi'| S\psi\rangle = \langle \psi'| X_{n,t}^\mu \psi \rangle$
 when $\psi,\psi' \in {\cal S}({\cal H})$. Since this space is dense, we have $ S\psi= X_{n,t}^\mu \psi$. To conclude the proof of the initial statement in (1), it is therefore sufficient to show that (\ref{rest0}) is valid. Properties (a)-(e) are then obvious consequences of the analogs for $ N_{n,t}^\mu$ and of the fact that ${\cal S}({\cal H})$ is also invariant under $U$, $ N_{n,t}^\beta$, and $P_{n\alpha}$.  The proof of (\ref{rest0}), taking  (\ref{TPOVM2}) into account, just amounts to prove that 
 \beq\nonumber
  \eta^{\mu\nu} \int_{\Sigma_{t,n}} x^k d\left\langle \frac{P_{n\mu}}{H_n}\psi \left|  \sQ_{n,t}(x) \frac{P_{n\nu}}{H_n} \right.\psi \right\rangle  +  \int_{\Sigma_{t,n}} x^k d \left\langle \frac{m}{H_n}\psi \left| \sQ_{n,t}(x)  \frac{m}{H_n}\right. \psi \right\rangle =0\:,
\eeq
 if $\psi \in {\cal S}({\cal H})$ and $k=1,2,3$. The case $k=0$ is trivial since in that situation $x^0=t$ can be extracted by the two integrals and the identity boils down to the trivial one 
 $\langle \psi | (H_n^{-1}(P_{n}^\mu P_{n\mu} + m^2 I) \psi\rangle=0$. Regarding the cases $k=1,2,3$, taking advantage of the spectral decomposition of $ N^k_{n,t}$, the identity above can be re-written 
  \beq \nonumber
  \eta^{\mu\nu} \left\langle \frac{P_{n\mu}}{H_n}\psi \left|   N^k_{n,t} \frac{P_{n\nu}}{H_n} \right.\psi \right\rangle  + \left\langle \frac{m}{H_n}\psi \left|  N^{k}_{n,t}  \frac{m}{H_n}\right. \psi \right\rangle =0\:,
\eeq
where we have also used the fact that ${\cal S}({\cal H}) \subset D( N^k_{n,t})$ and the former space is invariant under the selfadjoint  bounded operators $H_n^{-1}$ and $H_n^{-1}P_{n\mu}$ as the reader immediately proves. The identity above can be re-arranged to the equivalent form (remember that $H_n= -P_{n0}$)
  $$
  \eta^{\mu\nu} \left\langle\psi \left|  \frac{P_{n\mu}}{H_n}  \frac{P_{n\nu}}{H_n}   N^k_{n,t} \right.\psi \right\rangle  + \left\langle\psi \left|   \frac{m}{H_n} \frac{m}{H_n}  N^{k}_{n,t}\right. \psi \right\rangle $$
   $$+    \eta^{\mu\nu} \left\langle\psi \left|  \frac{P_{n\mu}}{H_n} \left [ N^k_{n,t}, \frac{P_{n\nu}}{H_n}\right] \right.\psi \right\rangle  + \left\langle\psi \left|  \frac{m}{H_n} \left[ N^{k}_{n,t},  \frac{m}{H_n}\right]\right. \psi \right\rangle =0\:. $$
   Representing the identiy above in the Hilbert space $L^2(\bR^3, d^3p)$ where ${\cal S}({\cal H})$ is represented by $\cS(\bR^3)$ itself, $P_{n\mu}= p_\mu \cdot$, $H_n= E_n(p) \cdot$ are multiplicative  and, for $\psi \in \cS(\bR^3)$ we have $ N^{k}_{n,t}\psi = i\frac{\partial}{\partial p_k}\psi$, 
   we see that the two commutators are multiplicative operators as well. Therefore, for instance $ \frac{P_{n\mu}}{H_n} \left[ N^k_{n,t}, \frac{P_{n\nu}}{H_n}\right] = \frac{1}{2}
    \frac{P_{n\mu}}{H_n} \left[ N^k_{n,t}, \frac{P_{n\nu}}{H_n}\right] +  \frac{1}{2}
    \left[ N^k_{n,t}, \frac{P_{n\nu}}{H_n}\right]  \frac{P_{n\mu}}{H_n} =   \frac{1}{2} \left[ N^k_{n,t}, \frac{P^2_{n\nu}}{H^2_n}\right]$ and similarly for the other addends.
 In summary,  the  indentity we need to establish can be re-arranged  to
    $$
  \left\langle\psi \left|  \frac{ \eta^{\mu\nu} P_{n\mu} P_{n\nu} +m^2I}{H^2_n}   N^k_{n,t} \right.\psi \right\rangle  +  \frac{1}{2}  \left\langle\psi \left| \left[ N^k, \frac{ \eta^{\mu\nu} P_{n\mu} P_{n\nu} +m^2I}{H^2_n}\right] \right.\psi \right\rangle  =0\:. $$
  which is evidently true, because $ \eta^{\mu\nu} P_{n\mu} P_{n\nu} +m^2I=0$ on ${\cal S}({\cal H})$, and it complete the proof of  (1).\\
  Let us pass to (2) and we prove (\ref{Xsquare}). With the same procedure used to prove (1) and if $\psi \in {\cal S}({\cal H})$, we find through (\ref{TPOVM2})
  $$ \int_{\Sigma_{n,t}} (x^k)^2 d\langle \psi| \sA_{n,t}(x) \psi \rangle = \langle \psi| ( N_{n,t})^2 \psi\rangle$$
  $$+\frac{1}{2} \eta^{\mu\nu} \left\langle\psi \left|  \frac{P_{n\mu}}{H_n} ( N^k_{n,t})^2 \frac{P_{n\nu}}{H_n} \right.\psi \right\rangle  + \frac{1}{2}\left\langle\psi \left|  \frac{m}{H_n} ( N^{k}_{n,t})^2  \frac{m}{H_n}\right. \psi \right\rangle\:.$$
The second line can be re-arranged to 
 $$\frac{1}{2}\eta^{\mu\nu} \left\langle\psi \left|  \frac{P_{n\mu}}{H_n}\frac{P_{n\nu}}{H_n} ( N^k_{n,t})^2  \right.\psi \right\rangle  +\frac{1}{2}  \left\langle\psi \left|  \frac{m}{H_n}\frac{m}{H_n} ( N^{k}_{n,t})^2  \right. \psi \right\rangle$$
 $$+\frac{1}{2} \eta^{\mu\nu} \left\langle\psi \left|  \frac{P_{n\mu}}{H_n} \left[( N^k_{n,t})^2, \frac{P_{n\nu}}{H_n}\right]  \right.\psi \right\rangle  +\frac{1}{2} \left\langle\psi \left|\frac{m}{H_n} \left[( N^{k}_{n,t})^2,  \frac{m}{H_n}\right]  \right. \psi \right\rangle\:. $$
 The first line vanishes, while the second can be explicitly computed by working in the space $L^2(\bR^3, d^3p)$ exactly as we did for item (1) and it becomes
 $$ -\eta^{\mu\nu}\frac{1}{2} \left\langle\psi \left|  \frac{p_{\mu}}{p_0} \left[\left(\frac{\partial}{\partial p_k}\right)^2, \spa\frac{p_{\nu}}{p_0}\right]  \right.\psi \right\rangle  -\frac{1}{2} \left\langle\psi \left|\frac{m}{p_0} \left[\left(\frac{\partial}{\partial p_k}\right)^2, \spa \frac{m}{p_0}\right]  \right. \psi \right\rangle $$
$$=\left\langle\psi \left|  \frac{1}{2p^0} \left(\partial_{p_k}\frac{p^k}{p^0} \right) \right.\psi \right\rangle  = \left\langle \psi\left| \frac{H_n^2-P_k^2}{2H_n^4}\psi\right. \right\rangle \:,$$
 where $p_0= -\sqrt{m^2 + \sum_{k=1}^3 p_k^2}$ and the operators $p_\mu$ being multiplicative. The proof of (\ref{Xsquare}) is over. To prove (\ref{HPN}), observe that
 $$(\Delta_\psi X^k_{n,t})^2 =   \int_{\Sigma_{n,t}} (x^k)^2 d\langle \psi| \sA_{n,t}(x) \psi \rangle -  \left( \int_{\Sigma_{n,t}} x^k d\langle \psi| \sA_{n,t}(x) \psi \rangle\right)^2$$
 $$=  \langle \psi| ( N^\mu_{n,t})^2\psi\rangle + \left\langle \psi\left| \frac{H_n^2-P_k^2}{2H_n^4}\psi\right. \right\rangle -\langle \psi |  N^k_{n,t}\psi\rangle^2 = (\Delta_\psi  N^k_{n,t})^2 +
   \left\langle \psi\left|\frac{H_n^2-P_k^2}{2H_n^4}\psi\right. \right\rangle\:.
 $$ By multiplying both sides with $(\Delta_\psi P_k)^2$ and taking advantage of the standard Heisenberg inequality, we get  (\ref{HPN}).
\end{proof}

\begin{remark}
{\em \begin{itemize}
 \item[(1)] The first-moment operator can be formally written within the QFT setting of Sect. \ref{secINTROT},
 $$X^k_{n,0} = \frac{1}{\sqrt{H_n}}  P_1 \int_{\Sigma_{n,0}} x^k :\spa\hat{T}_{\mu\nu}\spa:\spa(x) n^\mu n^\nu\:   d \Sigma_{nt}(x) P_1 \frac{1}{\sqrt{H_n}}\:.$$
 The internal integral is nothing but the $k$-component of the  {\em  boost generator}  in QFT evaluated at $t=0$.  The position  operator obtained in that way coincides with the known {\em  Born-Infeld position operator} 
  as discussed in \cite{BB} and remarked in \cite{Terno}.
\item[(2)]   Item (2) is of mathematical interest. If the identity were $$ \int_{\Sigma_{nt}} (x^k)^2 d\langle \psi| \sA_{n,t}(x) \psi \rangle =  \langle \psi| (X^k_{n,t})^2\psi\rangle\:,$$ since $X^k_{n,t}$ is symmetric and (\ref{chimoment}) is true, one could apply a known theorem by Naimark about the decomposition of symmetric operators in terms of POVMs (see Theorem 23  in \cite{DM} and the discussion about it). On account of that theorem, the POVM that decomposes $X^k_{n,t}$ according to (\ref{chimoment}) would be uniquely determined by its first moment  $X^k_{n,t}$, provided this operator be maximally symmetric on its domain, and it is our  case  since $X^k_{n,t}$ is  essentially self adjoint. Along this argument one would conclude that $\sA_{nt}= \sQ_{nt}$, since the latter  POVM (actually a PVM) decomposes $\overline{X^k_{n,t}} = N^k_{n,t}$
(as in (\ref{chimoment}) on ${\cal S}({\cal H})$)   in view of the spectral theorem.  In summary, the cumbersome addend to the right-hand side of  (\ref{Xsquare}) is responsible for the   failure 
of $\sA_{nt}= \sQ_{nt}$.
\item[(3)] Given a pure state represented by a unit vector $\psi \in {\cal S}({\cal H})$,   also the {\em standard Heisenberg inequalities}
$$\Delta_\psi  N^k_{n,t} \Delta_\psi P_{nk} \geq \hbar/2\:,$$
 are valid for $ N_{n,t}^k$ and $P_{nk}$ in addition to (\ref{HPN}),  as a consequence of the canonical commutation relations (\ref{CCR}). The point is that these relations refer to the physically wrong probability distribution, the one  constructed out of  the Newton-Wigner PVM $\sQ_{n,t}$  instead of the Terno POVM  $\sA_{n,t}$.  \hfill $\blacksquare$
\end{itemize}}
\end{remark}

\section{Every $n\in  \sT_+$ defines a causal time evolution for $\sA$}
This section is devoted to prove that every $n\in  \sT_+$ defines a {\em causal time evolution} in Castrigiano's sense,  according to (a)  in Definition \ref{LCR},
 for every family $\mu^\psi$
constructed out of  the  POVMs $\sA$ and a pure state $\psi \in {\cal H}$: $\mu^{\psi}_{n,t}(\Delta) := \langle \psi |\sA_{n,t}(\Delta)\psi\rangle$.  \\

\begin{remark} {\em There are other notions of spatial localization which are causal with respect to time evolution. The  localization in terms of POVMs  due to  Petzold et al.\cite{A,B} and Henning, Wolf \cite{C} are causal with respect to time evolution. The proof in \cite{B} can be made rigorous by means  of the  mathematical approach  developed  in this section.} $\hfill \blacksquare$
\end{remark}

\subsection{The heuristic idea  of a conserved  probability four-current}  The technology I will  exploit to prove that $\sA_{n,t}$ produces a family of probability measures that satisfies the requirement (a) in Definition \ref{LCR} for every $n\in  \sT_+$  is based on a probability four-current associated to  $\langle \psi |\sA_{n,t}(\Delta)\psi\rangle$. As explicitely observed in \cite{Terno}, (I disregard here   a number of mathematical details which will be fixed later)
$$\int_{\Delta} d\langle \psi |\sA_{n,t}(x)\psi\rangle = \int_{\Delta} J^{\psi}_{n\mu}(x) n^\mu d\Sigma_{n,t}(x)\:,$$  where $J^{\psi}_{n}$
satisfies a conservation equation $\partial^\mu J^{\psi}_{n \mu} =0$.  
The existence of such four current of probability was postulated in the general case in \cite{Jancewicz} and see also  \cite{A,B,C,D} for the use of similar currents in relation to the causality problem for massive Klein Gordon particles. A similar current exists for Dirac and Weyl particles \cite{Castrigiano1,Castrigiano2}.
Assuming that  $J^{\psi}_{n}$ is causal,  the divergence theorem should imply the validity of the local-causality requirement  when restricting to the family of $t$-parametrized rest spaces of a {\em unique reference frame}. I will  prove that it is the case in full generality, referring to every Lebesgue set $\Delta$. The extension to the full family of reference frames, i.e., the proof of the validity of  (b) in Definition \ref{LCR}, is not so easy since $J^{\psi}_{n}(x)$ itself  {\em depends on $n$} and one has to compare $\int_{\Delta} J^{\psi}_{n\mu}(x) n^\mu d\Sigma_{n,t}(x)$ and $\int_{\Delta} J^{\psi}_{n'\mu}(x) {n'}^\mu d\Sigma_{n',t'}(x)$.

\subsection{The probability current and its flow}
The first step of the  proof consists of explicitly writing  down  the current $J^\psi_n$ \cite{Terno} for the special case $\psi \in {\cal S}({\cal H})$. 
As  usual, I represent events by means of four-vectors $\bM \ni e= o+ x(e)$ where $\psi\in \sV$.

Directly from (\ref{TPOVM}), one has that, if $n\in  \sT_+$, $t\in \bR$, $\Delta \in \cL(\Sigma_{n,t})$, $\psi\in {\cal S}({\cal H})$
\beq\label{48}
\langle \psi |A_{n,t}(\Delta)\psi\rangle = \int_{\Delta} T^{\psi}_{\mu\nu}(x)_n n^\mu n^\nu d \Sigma_{n,t}(x)\:,
\eeq
where I introduced the coordinate representation of the {\bf stress-enegy tensor} of $\Phi^\psi_n$, 
$$T^{\psi}_{\mu\nu}(x)_n := \frac{1}{2}\left(\partial_\mu \overline{\Phi^\psi_n(x)}\partial_\nu\Phi^\psi_n(x) 
 +\partial_\mu \Phi^\psi_n(x)\partial_\nu\overline{\Phi^\psi_n(x)}\right)$$ \beq - \frac{1}{2}\eta_{\mu\nu} \left( \partial^\alpha \overline{\Phi^\psi_n(x)} \partial_\alpha \Phi^\psi_n(x) + m^2 \overline{\Phi^\psi_n(x)} \Phi^\psi_n(x) \right)\:,\label{TTT}
\eeq
associated to the smooth complex Klein-Gordon field
\beq
\Phi^\psi_n(x) := 
\int_{\sV_{m,+}}  \frac{\psi(p)e^{i p\cdot x} }{(2\pi)^{3/2}\sqrt{E_n(p)}} d\mu_m(p)\:, \label{wavePHI}
\eeq
Notice the further  factor $E^{-1/2}_n(p)$ when comparing with (\ref{wave}) which arises from the analogous factors in the right-hand side of (\ref{Ternoinformal}). Let us fix a Minkowskian coordinate system $t=x^0,x^1,x^2,x^3$ comoving with some $n\in  \sT_+$. Since the factor of $e^{i p\cdot x} $ in the integrand stays in $\cS(\bR^3)$, the function $\bR^3 \ni \vec{x} \mapsto \Phi^\psi_n(t,\vec{x})$ belongs to $\cS(\bR^3)$ as well for every $t\in \bR$. \\

\begin{definition} If $\psi \in {\cal S}({\cal H})$, $||\psi||=1$ and $n\in  \sT_+$, the associated {\bf probability four-current} of $\sA$  is the contravariant vector field $J^\psi_n$ on $\bM$ written in coordinates reads
\beq
J^{\psi\mu}_{n}(x) := n^\nu T^{\psi\mu}_{\nu}(x)_n \:,\label{53}
\eeq
where $(T^\psi_{\nu\mu})_n$ is defined in (\ref{TTT}). \\
\end{definition}

It is evident that, if $\psi \in {\cal S}({\cal H})$, $n\in  \sT_+$,  $t\in \bR$, and $\Delta \in \cL(\Sigma_{n,t})$, (\ref{48}) yields
\beq
\langle \psi|A_{n,t}(\Delta) \psi\rangle = \int_{\Delta} J^\psi_{n\mu}(x) n^\mu d\Sigma_{n,t}(x) \:.
\eeq


\begin{proposition}\label{PROPBAST1} If $\psi \in {\cal S}({\cal H})$, $n\in  \sT_+$, then  $J^{\psi}_{n}$  is either  the zero vector or is causal and past-directed. More precisely:
\begin{itemize}
\item[(1)] there is an open dense set $\sO^\psi_{n}\subset \bM$ where  $J^\psi_{n}$ is timelike and past-directed; 
\item[(2)]  if $e\in \bM \setminus \sO^\psi_n$, then either $J^\psi_{n}(e)=0$ or $J^\psi_{n}(e)$ is lightlike and past-directed;
\item[(3)] it holds $\sO^\psi_n =\{e \in \bM \:|\; \Phi^\psi_n(e) \neq 0\}$.
\end{itemize}
\end{proposition}

\begin{proof} We need some preparatory identities and inequalities.  Consider a Minkowskian coordinate system co-moving with $n$, so that $n^\mu = \delta^\mu_0$ and, if $\Phi^\psi_n = A_1+iA_2$ with $A_i$ real,
One can write
$$J^\psi_{n\mu} = J^\psi_{1n\mu}+J^\psi_{2n\mu}$$ where, for $j=1,2$,
$$J^\psi_{jn0} =   \frac{1}{2}\left(\partial_0A_j\partial_0 A_j 
+ \sum_{k=1}^3\partial_kA_j\partial_k A_j + m^2 A_j^2\right)\:, \quad 
J^\psi_{jnh} = \partial_0 A_j\partial_h A_j \:, \quad h=1,2,3.$$
At this juncture observe that, for $j=1,2$,
$$-g(J^\psi_{jn}, J^\psi_{jn})= \frac{1}{4}\left((\partial_0A_j)^2 
+ \sum_{k=1}^3(\partial_kA_j)^2 + m^2 A_j^2\right)^2 - \sum_{k=1}^3(\partial_k A_j \partial_0 A_j)^2$$
\beq= \frac{1}{4}\left((\partial_0A_j)^2 
- \sum_{k=1}^3(\partial_kA_j)^2 \right)^2 + \frac{1}{4} m^4A_j^4 + \frac{1}{2} m^2 (\partial_0 A_j)^2 A_j^2 + \frac{1}{2} m^2A_j^2 \sum_{k=1}^3(\partial_kA_j)^2 \geq 0\:.\label{INEQ}\eeq
Let us pass to prove (1).  Define $\sO^{\psi}_n$  as the set of events where  $J^\psi_{n\mu}$ is timelike. 
 Let us prove that the set $\sO^{\psi}_n$ is dense and open and the vectors in it are past-directed.

(Dense.) It is clear from the found inequality that, in particular,  if $\Phi^\psi_n(e) \neq 0$ then $J^\psi_{n\mu} = J^\psi_{1n\mu}+J^\psi_{2n\mu}$ is timelike so that  $e\in \sO^{\psi}_n$. 
 If  $x\in \bM$ and $N \ni x$ is an open neighborhood of it, suppose that  there is no $e\in N$
 where  $\Phi^\psi_n(e)\neq 0$. In particular $\Phi^\psi_n(e)= 0$  in the open spatial set $\Sigma_{n,t(x)}\cap N$. As a consequence, the spatial derivatives of $\Phi^\psi_n$ also vanishes on  $\Sigma_{n,t(x)}\cap N$ and (\ref{INEQ}) produces $-g(J^\psi_{jn}, J^\psi_{jn})= \frac{1}{4}(\partial_t A_j(e) )^4$. If the right-hand side vanished for all $e\in \Sigma_{n,t(x)}\cap N$ and $j=1,2$,  we would have that $\Phi^\psi_n(t,\cdot)$ and $(\overline{-\Delta +m^2})^{1/2}\Phi^\psi_n(t,\cdot) = -i \partial_t \Phi^\psi_n(t,\cdot)=0$ on that open set in $\Sigma_{n,t(x)}$.
 On account of Theorem \ref{teorem11}, we would have  $\Phi^\psi_n(t,\cdot)=0$ and thus $\psi=0$ by inverting (\ref{wavePHI}) and this is not allowed by hypothesis. We conclude that either $\Phi^\psi_n(t,e) \neq 0$ for some $e\in \Sigma_{n,t(x)}\cap N$ or $\Phi^\psi_n(t,e) = 0$ for all $e\in \Sigma_{n,t(x)}\cap N$, but $\partial_t\Phi^\psi_n(t,e) \neq 0$ for some $e\in \Sigma_{n,t(x)}\cap N$. In both cases,  (\ref{INEQ}) implies that
  $J^\psi_{n}$ is timelike somewhere in the neighborhood  $N$ of $x$. 
 We have proved that the set  $\sO^\psi_n$ where  $J^\psi_n$ is timelike is dense.
 
(Open.)  $\sO^\psi_n$ is also the preimage of an open set (the open future cone) according to a continuous map and thus it is open as well. 

(Past directed.)  Since $n$ is future-directed and $J^\psi_{jn} \cdot n = J^\psi_{jn0}  \geq 0$, we also have that $J^\psi_n$ is past-directed when it does not vanish.\\
(2) Consider  $e\in \bM \setminus \sO^\psi_n$, namely $J^\psi_n(e)$ is not timelike. Since 
$J^\psi_{n} = J^\psi_{1n}+J^\psi_{2n}$ we have
$$g(J^\psi_{n},J^\psi_{n})= g(J^\psi_{1n},J^\psi_{1n}) + g(J^\psi_{2n},J^\psi_{2n}) + 2g(J^\psi_{1n},J^\psi_{2n})\:.$$
Notice that all scalar products taking place on the right-hand side  above  are non-positive: the first two because of  (\ref{INEQ}) and the last one because  
the two vectors are the limit of past directed timelike vectors for (1).
 Since the left-hand side is zero by hypothesis, we have the following two  possibilities.
 $J^\psi_n(e)$ vanishes (if both  $J^\psi_{1n}$ and  $J^\psi_{2n}$
vanish) or it is light like (if one of the two vanishes and the other is lightlike or if both are lightlike and parallel). In all these cases
both $A_1$ and $A_2$ vanish on account of (\ref{INEQ}) where $m>0$,  so that  $\Phi^\psi_n(e)=0$ as well.  
To conclude, observe that if $J^\psi_n$ is lightlike, then  it must be past-directed by continuity because $\sO^\psi_n$ is dense and the vectors in that set are past-directed. The proof of (3) has been given while establishing (1) and (2).
\end{proof}

\subsection{Every $n\in  \sT_+$ defines a causal time evolution for  $\sA$}

First of all, observe that  if $D\subset \Sigma_{n,t_1}$ is an open ball, then $J^\pm(D)$ are open as well as it arises per direct inspection. This immediately  implies that $J^\pm(\Delta_1)$ are open if $\Delta_1 \subset \Sigma_{n,t_1}$ 
is open  and non-empty. As a consequence, 
when $\Delta_1 \subset \Sigma_{n,t_1}$ is open,
the intersections $J^\pm(\Delta_1) \cap \Sigma_{n',t'}$ are open as well in the relative topology. I will use this fact several times in the rest of the paper.\\

\begin{lemma}\label{LEMMAA}  Consider the spatial localization observable $\sA$. Take $n\in  \sT_+$ and $t_1,t_2 \in \bR$ with $t_2 \neq t_1$. Let  $\Delta_1 \subset \Sigma_{n,t_1}$  be a finite  union of  non-empty open balls with finite radius, and let $\Delta_2 := (J^+(\Delta_1) \cup J^-(\Delta_1))\cap \Sigma_{n,t_2}$ be the  corresponding open  set in $\Sigma_{n,t_2}$. Then
\beq
\langle \psi | \sA_{n,t_1}(\Delta_1) \psi\rangle \leq  \langle \psi | \sA_{n,t_2}(\Delta_2) \psi\rangle \label{nonsharp}
\eeq 
is valid for every $\psi \in {\cal S}({\cal H})$ with $||\psi||=1$.
\end{lemma}

\begin{proof}
 As a  first case, we assume that  $\Delta_1 \subset \Sigma_{n,t_1}$ is an open ball of finite radius, so that  $\Delta_2$ in $\Sigma_{n,t_2}$ is an  analogous open set in  $\Sigma_{n,t_2}$. 
 Let us suppose $t_2>t_1$ (the other case is analogous) and  consider $B \subset \bM$ whose boundary is made of the two bases $\Delta_1$, $\Delta_2$, and the portion $L$ of $\partial J^+(\Delta_1)$ between them.  $B$ is a manifold with boundary and we can use the Stokes-Poincar\'e theorem for the 3-forms\footnote{One cannot take advantage of the vector field version of the theorem because the portion $L$ of the boundary has a  degenerated induced metric.}
$$\nu^\psi_n =  -\frac{1}{3!} \sqrt{-\det(g)} \epsilon_{\alpha\beta\gamma\delta} J^{\psi \delta}_n dx^\alpha \wedge dx^\beta \wedge dx^\gamma $$
 associated 
to the current  $J^\psi_n$ for the considered $\psi \in {\cal S}({\cal H})$. We have chosen a Minkowskian coordinate system $t=x^0,x^1,x^2,x^3$ comoving with $n$ to write down  the components of $\nu^\psi_n$ as above. With the choices above, the integral of the form on $\Delta_{t_2}$ gives $$\int_{\Delta_{2}}  \nu^\psi_n = \int_{\Delta_{2}} J^\psi_n \cdot n d\Sigma_{n,t_2}=  \langle \psi | \sA_{n,t_2}(\Delta_2) \psi\rangle\:.$$
 Since $J^\psi_n$ is conserved, the integral of $\nu^\psi_n$ on $B$ vanishes, so that,  
\beq  \langle \psi | \sA_{n,t_2}(\Delta_2) \psi\rangle -  \langle \psi | \sA_{n,t_1}(\Delta_2) \psi\rangle = \int_L \nu^\psi_n\:.\label{intinT}\eeq
To compute the integral we change coordinates and we pass to a system of lightlike and polar coordinates $u,v, \theta, \phi$
where $r, \theta,\phi$ are standard polar spherical coordinates in $\Sigma_{n,t_1}$ with  center given by the center of $\Delta_1$ and $u:= t+r$, $v:= t-r$ so that $u$ is a lightlike future increasing coordinate along $L$.  
With these coordinates,
$$g = -\frac{1}{2}du \otimes dv - \frac{1}{2} dv\otimes dv + \frac{1}{4}(u-v)^2 (\sin^2 \theta d\phi \otimes d\phi + d\theta \otimes d\theta)$$
and, writing $J$ for $J_n^\psi$,
\beq  \nu^\psi_n = -\frac{1}{2} (u-v)^2 \sin \theta  J^v du \wedge d\theta \wedge d\phi\:.  \label{NU}\eeq
Now, observe that, since $J^\psi_n$ is  past directed (if it does not vanish), we must have $2J^t = J^u+ J^v  \leq 0$. The condition that $J^\psi_n$ is zero or causal reads
$$-J^uJ^v + h(\vec{J}, \vec{J}) \leq 0\:,$$
where $h$ is the Euclidean metric on $\Sigma_{n,t}$ and $\vec{J}$ the spatial part of $J_n^\psi$. In summary, $J^uJ^v \geq 0$ and $J^u+J^v \leq 0$, so that  $J^v,J^u \leq 0$.
Since $\theta \in [0,\pi]$ in (\ref{NU}) and $v=0$ on $L$, we conclude that
\beq  \int_L \nu^\psi_n = -\int_L  \frac{1}{2}u^2 \sin \theta  J^v du \wedge d\theta \wedge d\phi \geq 0\:.\label{CONT}\eeq
Up to now we have established that 
\beq \langle \psi | \sA_{n,t_1}(\Delta_1) \psi\rangle \leq  \langle \psi | \sA_{n,t_2}(\Delta_2) \psi\rangle\:.\label{INEQJ}\eeq
 To conclude the proof it is sufficient to observe what follows  in the case  $\Delta_1$ is a finite union of finite-radius open balls $\Delta^{(j)}_{1}$, 
 $j=1,\ldots, N$. We can always assume that  no  ball of the family is a subset of another ball of the family.
  Since $N$ is finite,   the region of  $\partial J^+(\Delta_1)$ 
 between $t_1$ and $t_2$  is a piecewise smooth lightlike submanifold and we can apply the above reasoning by changing coordinates for every cone of the family. The integral over the surface $\partial J^+(\Delta_1)$ 
 between $t_1$ and $t_2$  is a finite sum of contributions of type (\ref{CONT}) where each integral is now performed on a smaller  portion of each conical surface. However each contribution is non-negative because the integrated function is non-negative.
\end{proof}

\begin{remark}{\em  Even if it is not strictly necessary for our final goal, I prove  that, if restricting to a suitable dense subspace of ${\cal S}({\cal H})$, the inequality in (\ref{nonsharp}) can be made sharp.
I  consider a subspace ${\cal D}({\cal H}) \subset {\cal S}({\cal H})$ of  vectors  $\psi \in \cal H$ such that  there is $n\in  \sT_+$ and  a Minkowski coordinate system co-moving  with $n$ such that 
$\bR^3 \ni \vec{p} \mapsto \psi(E_n(p), \vec{p}_n) \in {\cD}(\bR^3)$ (the test-function  space on $\bR^3$) when  represented in the spatial coordinates on $\bR^3$. 
 The definition of ${\cal D}({\cal H})$ does not depend of the choice of  $n$ and co-moving Minkowskian coordinates as  ${\cal D}({\cal H})$ is  invariant under the representation $U$ of $IO(1,3)_+$ in (\ref{ACT1}). Finally, ${\cal D}({\cal H}) \subset {\cal S}({\cal D})$  is  dense in ${\cal H}$.
 The proof of these   elementary facts  is analogous to the one of  ${\cal S}({\cal H})$ and it is left to the reader.
 
 Relying on the  the  well posedness of the {\em Characteristic Cauchy problem} on Lorentzian cones, the following precise result is valid.\\

\begin{proposition}\label{PROPSHARP} With the hypotheses of Lemma \ref{LEMMAA}, if $\psi \in {\cal D}({\cal H})$ with $||\psi||=1$, then  inequality (\ref{nonsharp}) holds in the sharpest form
\beq
\langle \psi | \sA_{n,t_1}(\Delta_1) \psi\rangle <  \langle \psi | \sA_{n,t_2}(\Delta_2) \psi\rangle \label{sharp}
\eeq 
\end{proposition}

\begin{proof}  See Appendix \ref{APPPROOFS}.
\end{proof}  \hfill $\blacksquare$}
\end{remark}

\noindent I come back to the main stream of the  reasoning with a second lemma.\\

\begin{lemma} \label{LEMMAB} Consider the spatial localization observable $\sA$. Take $n\in  \sT_+$ and $t_1,t_2 \in \bR$ with $t_2 \neq t_1$. Let  $\Delta_1 \subset \Sigma_{n,t_1}$ be an non-empty open set (respectively a compact set), and let $\Delta_2 := (J^+(\Delta_1) \cup J^-(\Delta_1))\cap \Sigma_{n,t_2}$ be the  corresponding open (resp. compact)  set in $\Sigma_{n,t_2}$. Then
\beq
\langle \psi | \sA_{n,t_1}(\Delta_1) \psi\rangle \leq  \langle \psi | \sA_{n,t_2}(\Delta_2) \psi\rangle
\eeq 
is valid for every $\psi \in {\cal S}({\cal H})$ with $||\psi||=1$.
\end{lemma}

\begin{proof} We always assume $t_2> t_1$, since the other case has a similar proof.  First of all, we already know  that if $\Delta_1$ is open then $\Delta_2$ is open as well. The case of $\Delta_1$ compact is a subcase of a known fact valid in globally hyperbolic spacetimes (like $\bM$): if $K$ is compact, the intersection of $J^+(K)$ and a spacelike  Cauchy surface (like $\Sigma_{n,t_2}$)  is compact as well. \\
Let us first examine the case of $\Delta_1\subset \Sigma_{n,t_1}$ open. According to Theorem  1.26 in \cite{EG},  for every $\delta>0$, there exist a countable collection $\{\Gamma_j\}_{j=1,2,\ldots }$  of disjoint (non-empty) closed balls $\Gamma_j \subset \Delta_1$ with diameter less than  $\delta$, such that \beq\int_{\Delta_1 \setminus \bigcup_{j\in \bN} \Gamma_j} 1\: d\Sigma_{n,1} =0 \label{GAMMA}\eeq
 where we remind the reader that $d\Sigma_{n,1}$ is the Lebesgue measure when written in the spatial Minkowskian coordinates comoving with $n$. Evidently we can assume that the balls are open (and their closures are disjoint) since $\partial \Gamma_j$ has zero Lebesgue measure.
 Let us define $\Delta_1':=  \bigcup_{j\in \bN} \Gamma_j$ 
 and $\Delta'_2 := \Sigma_{n,t_2} \cap J^+(\Delta'_1)$.
 Since the probability  measure defined by $\sA_{n,t_1}$  and $\psi$ is per definition absolutely continuous with respect to the Lebesgue measure,
 (\ref{GAMMA}) yields $\langle \psi| \sA_{n,t_1}(\Delta'_1)\psi \rangle = \langle \psi| \sA_{n,t_1}(\Delta_1)\psi \rangle \in [0,+\infty]$. Furthermore, since $\Delta'_1 \subset \Delta_1$,  it must be  $\Delta_2' \subset \Delta_2$ and thus $\langle \psi| \sA_{n,t_2}(\Delta'_2)\psi \rangle \leq
\langle \psi| \sA_{n,t_2}(\Delta_2)\psi \rangle$. In summary, to prove the thesis, it is sufficient to establish that
$\langle \psi| \sA_{n,t_1}(\Delta'_1)\psi \rangle \leq
\langle \psi| \sA_{n,t_2}(\Delta'_2)\psi \rangle$.
Let us define $\Delta_1^N := \cup_{j=1}^N \Gamma_j$ and  $\Delta_2^N := J^{+}(\Delta_1^N) \cap \Sigma_{n,t_2}$. By additivity and taking Lemma \ref{LEMMAA} into account, $$\langle \psi| \sA_{n,t_1}(\Delta'_1)\psi \rangle = \lim_{N\to +\infty} \langle \psi| \sA_{n,t_1}(\Delta_1^N)\psi \rangle \leq \lim_{N\to +\infty}  \langle \psi| \sA_{n,t_2}(\Delta_2^N)\psi\rangle \leq  \langle \psi| \sA_{n,t_2}(\Delta'_2)\psi\rangle\:. $$ Notice that the limit of the right-most side exists because the sequence is non-decreasing as $\Delta_2^N \subset \Delta_2^{N+1} \subset \Delta_2'$ by construction.\\
Let us pass to prove the thesis for $\Delta_1$ compact. Since $\Sigma_{n,t_1}$ is a metric space and $\Delta_1$ compact, it is not difficult to construct a sequence 
of open sets $A_1 \supset A_2 \supset \cdots \supset \Delta_1$ such that 
 $$\Delta_1 = \bigcap_{j=1,2,\ldots} A_j\:.$$
Each  $A_j$ is the union of a finite (but arbitrarily large)  number of balls centered on some points of $\Delta_1$ with radius less  than  $\delta_j \to 0^+$.
 As a consequence  $$\Delta_2 = \left(\bigcap_{j=1,2, \ldots} J^+(A_j)\right) \cap \Sigma_{n,t_2}\:. $$
The inclusion $\subset$ immediately arises from the definitions, the other inclusion is less trivial. Let us prove it.
If $e$ belongs to the right-hand side of the identity above and, as said,  $A_j$ is the  finite union of balls of radius $\delta_j>0$ centered on some points of $\Delta_1$, we have that\footnote{This is valid if $\Sigma_{n,t_1}$ and $\Sigma_{n',t_2}$ are parallel as it is since we are assuming $n=n'$. However a similar argument is valid if $n\neq n'$, finding  $\mbox{dist}(e,  J^+(\Delta_1) \cap \Sigma_{n',t_2}) <\epsilon \delta_j$
 for some $\epsilon>0$ independent of $j$.} $\mbox{dist}(e, J^+(\Delta_1) \cap \Sigma_{n,t_2}) < \delta_j$  for every $\delta_j \to 0^+$. As a consequence $e$ is an accumulation point of $\Delta_2 = (J^+(\Delta_1) \cap \Sigma_{n,t_2}) \cap \Sigma_{n,t_2}$ which is compact, thus closed (the space being Hausdorff). Hence $e \in \Delta_2$. 
 Finally, taking advantage of the already proved result on open sets and internal continuity
$$\langle \psi | \sA_{n,t_2}(\Delta_2) \psi\rangle  = \inf_j \langle \psi | \sA_{n,t_2}(J^+(A_j) \cap \Sigma_{n,t_2} ) \psi\rangle
 \geq  \inf_j \langle \psi | \sA_{n,t_1}(A_j) \psi\rangle = \langle \psi | \sA_{n,t_1}(\Delta_1) \psi\rangle\:.$$
\end{proof} 

I am now in a position to  prove the main result of this section, that every $n\in  \sT_+$ defines a causal time evolution (according to (a) in  Definition \ref{LCR})  for every spatial localization probability measure  constructed out of the Terno POVM $\sA$ and every pure state $\psi \in {\cal H}$.\\

\begin{theorem}\label{TEO2}  Consider the spatial localization observable $\sA$. Take $n\in  \sT_+$ and $t_1,t_2 \in \bR$. Let  $\Delta_1 \subset \Sigma_{n,t_1}$ be a Lebesgue  set  and let $\Delta_2 := (J^+(\Delta_1) \cup J^-(\Delta_1))\cap \Sigma_{n,t_2}$ be the  corresponding   set in $\Sigma_{n,t_2}$. Then
\beq
\langle \psi | \sA_{n,t_1}(\Delta_1) \psi\rangle \leq  \langle \psi | \sA_{n,t_2}(\Delta_2) \psi\rangle\:, \quad \mbox{ $\forall \psi \in {\cal H}$ with $||\psi||=1$.}
\eeq  In other words, every $n\in  \sT_+$ defines a causal time evolution according to   (a) in  Definition \ref{LCR} for the family of spatial localization probability measures  $\mu^\psi(\cdot) := \langle \psi | \sA(\cdot) \psi\rangle$.
\end{theorem}

\begin{proof}  First of all,  notice that $\mu^{\psi}_{n,t}(\cdot) := \langle \psi| \sA_{n,t}(\cdot) \psi \rangle$, for $\psi \in {\cal S}({\cal H})$ is necessarily regular when restricted to  $\cB(\Sigma_{n,t})$, since $\Sigma_{n,t}$  is countable union of compacts with finite measure (Theorem 2.18 in \cite{Rudin}). As a consequence the completion
$\overline{\mu^{\psi}_{n,t}|_{\cB(\Sigma_{n,t})}}$ of $\mu^{\psi}_{n,t}|_{\cB(\Sigma_{nt})}$ 
 is regular as well (Prop. 1.59 in \cite{Cohn}). The  $\sigma$-algebra of the regular complete measure $\overline{\mu^{\psi}_{n,t}|_{\cB(\Sigma_{nt})}}$  includes the Lebesgue $\sigma$-algebra in particular, and the completion $\overline{\mu^{\psi}_{n,t}|_{\cB(\Sigma_{nt})}}$  restricted to 
 $\cL(\Sigma_{n,t})$ coincides 
  to $\mu^{\psi}_{n,t}$ itself. This can be seen as follows. The   $\sigma$-algebra of a completion $\overline{\mu}$ --  where $\mu : \cS(X) \to [0,+\infty]$ is a positive $\sigma$-additive measure -- can be constructed as the family of sets $E\cup Z$ where $E\in \cS(X)$ and $Z\subset F \in \cS(X)$ with $\mu(F)=0$. Obviously $\overline{\mu}(E\cup Z) := \mu(E)$.
From these properties we can write, $\overline{\mu^{\psi}_{n,t}|_{\cB(\Sigma_{nt})}}(G)=\mu^{\psi}_{n,t}(G)$  if $G\subset \cL(\Sigma_{n,t})$ since $G= E\cup Z$ where
 $E\in \cB(\Sigma_{n,t})$ and $Z \subset F \in \cB(\Sigma_{n,t})$ such that $F$ has zero Lebesgue measure and thus $\mu^{\psi}_{n,t}(F)=0$ 
because  $\mu^{\psi}_{n,t}$ is absolutely continuous with respect to the Lebesgue measure.
   We conclude that $\mu^{\psi}_{n,t}$ is regular on the Lebesgue $\sigma$-algebra because it is the restriction of a regular measure. In particular it is inner regular.
So, if $\Delta_1$ is Lebesgue-measurable,   for $\psi \in {\cal S}({\cal H})$ we can take advantage of Lemma \ref{LEMMAB} proving that 
$$\langle \psi | \sA_{n,t_1}(\Delta_1) \psi\rangle  = \sup\{\langle \psi | \sA_{n,t_1}(K) \psi\rangle \:|\: K\subset \Delta_1\:, \mbox{$K$ compact} \}$$
$$\leq  \sup\{\langle \psi | \sA_{n,t_2}(J^+(K) \cap \Sigma_{n,2}) \psi\rangle \:|\: K\subset \Delta_1\:, \mbox{$K$ compact} \} \leq \langle \psi | \sA_{n,t_1}(\Delta_2) \psi\rangle \:,$$
where we have also used the fact that $J^+(K) \cap \Sigma_{n,2} \subset J^+(\Delta_1) \cap \Sigma_{n, 2} = \Delta_2$.\\
The thesis is therefore true if $\psi \in {\cal S}({\cal H})$ with $||\psi||=1$. Evidently the last requirement can be dropped by bi-linearity of the scalar product. Since ${\cal S}({\cal H})$
is dense in ${\cal H}$ and the scalar product is continuous, the result extends to the whole Hilbert space and the proof is over.
\end{proof}

\begin{corollary} There is no state $\psi \in {\cal H}$ that  satisfies the hypotheses of the Hegerfeldt theorem (Theorem \ref{HGT}) for any  family of bounded balls in the rest space of any arbitrarily fixed $n\in  \sT_+ $.
\end{corollary}

\begin{proof} The thesis of Hegerfeldt's theorem is incompatible with the result of the previous theorem. 
\end{proof}

\section{Subtleties with the notion of position and  Castrigiano's causality requirement}
There is a crucial feature of the notion of spatial position by Terno: it uses a four current of probability that, in spite of being a four-vector,  depends on the reference frame $n$  as it is evident in (\ref{53}) when $\psi \in {\cal S}({\cal H})$. That is an unavoidable fact since the notion of energy-momentum current  has the same type of dependence: $J^\nu_n = n^\mu {T_{\mu}}^\nu$. This feature leads to a more articulated picture where one can define the probability to find a particle in $\Delta \subset \Sigma_{n',t'}$ still referring to the current associated to $n\neq n'$!  That is permitted because 
$$J^{\psi\mu}_{n}(x)  n'_\mu \geq 0$$
in view of Proposition \ref{PROPBAST1}, when $n'\in \sT_+$. In fact $J^{\psi\mu}_{n}(x)$ is causal and past  directed or vanishes producing the inequality above just  because $n'$ is timelike and future directed. So that, if $\psi \in {\cal S}({\cal H})$, one can define a spatial localization probability
\beq
\mu^{\psi,n}_{n',t'}(\Delta)  := \int_{\Delta}   J^{\psi\mu}_{n}(x)  n'_\mu d\Sigma_{n',t'}\:,\quad   \Delta \in \cL(\Sigma_{n',t'})\:, \quad -n'\cdot x = t'\:.
\eeq
The divergence theorem, exploiting the fact that $ J^{\psi\mu}_{n}(x)$ rapidly vanishes at spatial infinity and that $\partial_\mu J^{\psi\mu}_{n}(x)= \partial_\mu n^\nu T^{\psi\mu}_{\nu}(x)_n=0$, assures the correct normalization
$$\mu^{\psi,n}_{n',t'}(\Sigma_{n',t'})  := \int_{\Sigma_{n',t'}}   J^{\psi\mu}_{n}(x)  n'_\mu d\Sigma_{n',t'} =
 \int_{\Sigma_{n,t}}   J^{\psi\mu}_{n}(x)  n_\mu d\Sigma_{n,t}= \langle \psi|\sA_{n',t'}(\Sigma_{n,t})\psi \rangle= 1\:.$$
Physically speaking, $\mu^{\psi,n}_{n',t'}(\Delta)$ accounts for  the probability to find a particle in $\Delta \subset \Sigma_{n',t'}$ {\em using detectors which are at rest in $n$ but synchronized with $n'$}. There is no reason why this probability should coincide with $\mu^\psi_{n',t'}(\Delta)=\langle \psi|\sA_{n',t'}(\Delta)\psi \rangle$ as the corresponding energy densities do not.  This result opens a new perspective on the notion of spatial localization  which deserves to be investigated.

Mathematically speaking all that can be encapsulated into a new family of POVMs depending on both $n$ and $n'$ (and $t'$).\\

\begin{theorem} If $n,n' \in \sT_+$ and $t'\in \bR$, there is only one POVM with effects $\sM^n_{n',t}(\Delta)  \in \gB({\cal H}) $ for
$\Delta \in \cL (\Sigma_{n',t'})$ such that
\beq\label{MMM}
\langle \psi | \sM^n_{n',t}(\Delta) \psi \rangle = \int_{\Delta}   J^{\psi\mu}_{n}(x)  n'_\mu d\Sigma_{n',t'}(x) \:\:, \forall \psi \in {\cal S}({\cal H})\:.
\eeq Furthermore the following holds.
\begin{itemize} \item[(1)] It has the form, in terms of the Newton-Wigner  POVM $\sQ_{n',t'}$ on $\Sigma_{n',t'}$,
$$
\sM^n_{n',t'}(\Delta) =\frac{1}{2}\left( \sqrt{\frac{H_{n'}}{H_{n}}}\sQ_{n',t'}(\Delta)
 \sqrt{\frac{H_{n}}{H_{n'}}} +  \sqrt{\frac{H_{n}}{H_{n'}}} \sQ_{n',t'}(\Delta) \sqrt{\frac{H_{n'}}{H_{n}}}\right)
$$
\beq 
-\frac{n\cdot n}{2}  \sqrt{\frac{H_{n'}}{H_{n}}}\left(\eta^{\mu\nu}\frac{P_{n\mu}}{H_{n'}}  \sQ_{n',t'}(\Delta) \frac{P_{n\nu}}{H_{n'}}  + \frac{m}{H_{n'}} \sQ_{n',t'}(\Delta)  \frac{m}{H_{n'}} \right)\sqrt{\frac{H_{n'}}{H_{n}}}\:.\label{MPOVM}
\eeq
(Where the various everywhere-defined bounded  composite operators $H_n/H_{n'}$ etc are defined in terms of the joint spectral measure of $P^\mu$ and standard spectral calculus).
\item[(2)]
It reduces to the Terno POVM  for $n=n'$:
\beq
\sM^n_{n,t}(\Delta) = \sA_{n,t}(\Delta)\:, \mbox{if $n\in \sT_+, t\in \bR$ and $\Delta \in \cL(\Sigma_{n,t})$.}
\eeq
 \item[(3)] The $IO(1,3)_+$ covariance relations are valid,
\beq 
 U_{h}  \sM^n_{n',t'}(\Delta) U_{h}^{-1} = \sM^{\Lambda_h n}_{\Lambda_h n', t'_h}(h\Delta) \:, \quad \forall \Delta \in \cL(\Sigma_{n',t'})\:, \quad \forall h \in IO(1,3)_+\label{ACOVNWT2}\:.
\eeq
\end{itemize}
\end{theorem}

\begin{proof} ({\em Initial statement and  (1)}). Let us call $F$ the operator defined by the right-hand side of (\ref{MPOVM}). It   is evidently everywhere defined and bounded on ${\cal H}$. By polarization and density of ${\cal S}({\cal H})$, it is completely determined by  the values 
$\langle \psi |F \psi \rangle$ when $\psi \in {\cal S}({\cal H})$. Let us prove that it satisfies (\ref{MMM}).
 Per direct inspection we have that, if $\psi\in {\cal S}(\cal H)$, taking (\ref{TTT}) and (\ref{wavePHI}) into account, the right-hand side  of (\ref{MMM}) can be written, with $-n'\cdot x= t'$
$$\int_{\sV_{m,+}}\sp\sp \sp \sp d\mu(p) \int_{\Delta}\sp\spa  d\Sigma_{n',t'}(x) \int_{V_{m,+}} \sp\sp\sp  d\mu(q) \frac{e^{-i(q-p)\cdot x}}{(2\pi)^3} \frac{E_n(p)E_{n'}(q)+ E_n(q) E_{n'}(p) - n\cdot n'(p^\alpha q_\alpha + m^2)}{2 \sqrt{ E_n(q) E_n(p)}} \overline{\psi(p)}\psi(q)$$
which, in turn,  coincides with $\langle \psi |F \psi\rangle$ when taking (\ref{firstPVM}) into account, as wanted. Notice that (\ref{MMM}) implies that  the everywhere defined extended operator $\sM^n_{n',t'}(\Delta)$
is positive as it is the continuous extension of a positive operator. The family of these operators, with  $n, n',t'$ fixed, is also weakly $\sigma$-additive in $\Delta$ because $\sQ_{n',t'}$ in the right-hand side of  (\ref{MPOVM}) is weakly $\sigma$-additive,  and the operators appearing as factors are bounded and everywhere defined. As the family $\sM^n_{n',t'}(\Delta)$, with $\Delta$ variable in $\cL(\Sigma_{n',t'})$,   is made of positive operators with $\sM^n_{n',t'}(\Sigma_{n',t'})=I$ (direct inspection), we conclude that the said family  (with $n$ fixed) is a (normalized) POVM on $\cL(\Sigma_{n',t'})$.\\
(2) It is obvious from (\ref{MPOVM}) and (\ref{TPOVM2}).\\
(3) The proof immediately arises from the analogous covariance properties of $\sQ_{n,t}$ and the basic covariance properties of $\sH_n$ and composite (bounded everywhere defined) operators $H_n/N_{n'}$, $m/H_{n}$, $P_{n'}^\mu/H_n$.\\
\end{proof}

\begin{remark} {\em For a given $n_0 \in \sT_+$,  the physical meaning of the family of POVMs $$\sM^{n_0}:= \{\sM^{n_0}_{n, t}\}_{n\in \sT_+, t\in \bR}$$ is the {\em notion of spatial position observable, referred to all reference frames $n\in \sT_+$ and every global  time $t\in \bR$ of each such reference frame,  when the used class of detectors is always co-moving with $n_0$}.}\\
\end{remark}

To conclude this work,  I  prove that for every given $n_0 \in \sT_+$,  the family of POVMs $\sM_{n_0}$ satisfies  Castrigiano's causality condition.\\

\begin{theorem}\label{LAST}
For given $n_0\in \sT_+$ and $\psi \in {\cal H}$, define the family of probability measures $\mu^{\psi, n_0}_{n,t}$
$$\mu^{\psi, n_0}_{n,t}(\Delta) := \langle \psi|\sM^{n_0}_{n,t}(\Delta) \psi\rangle\:, \quad n\in \sT_+, t\in \bR, \Delta \in \cL(\Sigma_{n,t})\:.$$
That family satisfies  Castrigiano's causality condition  (b)  in Definition \ref{LCR}. 
$$
\mu^{\psi, n_0}_{n,t}(\Delta) \leq \mu^{\psi, n_0}_{n',t'}(\Delta') \quad \forall n, n'\in  \sT_+ \:,  \forall  t, t'\in \bR\:, \forall \Delta \in \cL(\Sigma_{n,t}) 
$$
where  $\Delta' := \left(J^+(\Delta)  \cup J^-(\Delta) \right)\cap \Sigma_{n',t'}$.\\
In particular,  the  time evolution associated to every $n$ is causal according to (a) Definition \ref{LCR}.\\
\end{theorem}

\noindent {\em Sketch of proof}. Condition (a) in Definition \ref{LCR} is satisfied if condition (b) holds, so that  it  suffices to prove the validity of the latter.
The proof of Theorem \ref{TEO2} and its preparatory lemmata can be performed also for the considered case since the only relevant two facts, for $\psi \in {\cal S}({\cal H})$, are that (i)  the  values of $\mu^{\psi, n_0}_{n,t}(\Delta)$ and $\mu^{\psi, n_0}_{n',t'}(\Delta')$ -- where for the moment $n=n'=n_0$ -- are spatial boundary integrals of the conserved  four current $J^\psi_{n_0}$ and (ii) that $J^\psi_{n_0}$ is either zero or causal and past directed. These facts are valid also dropping the requirement $n=n'=n_0$. It does not matter if  the normal vectors $n$ and $n'$  to the two  hyperplanes containing respectively  $\Delta$ and $\Delta'$  are both parallel to the vector $n_0$ defining $J^\psi_{n_0}$ or not, so we can definitely drop the requirement $n=n'=n_0$. Indeed, in  proving  Theorem \ref{TEO2} the bases of the four-dimensional solid used to integrate the current were orthogonal to $n_0$ just as a contingent fact, due to   the very definition of the measures $\mu^\psi_{n,t}$ which is now relaxed. The only case where the above proof has to be slightly changed  is when the possible intersection of $\Sigma_{n,t}$ and $\Sigma_{n',t'}$ passes through $\Delta$. In that case it is convenient to treat separately the two parts of $\Delta$.  \hfill $\Box$

\section{Discussion} In this work, I rigorously  proved that, when referring to the only issue (I1) of the Introduction, a spatial notion of localization for a massive Klein Gordon particle is possible without problems with causality (with some caveat however, see below), avoiding the pathologies predicted by Hegerfeldt's theorem in particular. 
As is well known from long time, this latter obstruction prevents in particular the existence of  spatially localized  states. 
The crucal mathematical notion is here  the covariant family of POVMs $\sA$ proposed by Terno \cite{Terno} which has been analysed with a broad mathematical  detail, focusing  on its interplay with the popular Newton-Wigner notion of  spatial localization.
This analysis  showed that the notion of localization based on the POVM $\sA$ and the associated first moment in particular, keep many good properties of the Newton-Wigner localization notion while they drop  many problematic issues. 
 To what extent  this notion is compatible with the interpay of  causality and  post-measurement state (I2) was not the object of this work and it will be investigated elsewere.  
Terno's  notion seems  in good agreement with  Castrigiano's notion of causal evolution ((a) Definition \ref{LCR}).
The validity of the very Castrigiano  causality condition ((b) Definition \ref{LCR}) needs more care and a different, perhaps physically more subtle,  analysis than  the case of {\em causal systems} rigorously  treated by Castrigiano \cite{Castrigiano2}.  Terno's notion of spatial localization relies upon the notion of energy density and not upon the notion of density of charge. The former is associated to a conserved {\em tensor field}, the stress energy tensor $T_{\mu\nu}$, instead of a  vector field. 
As a matter of fact, the relevant probability density in the reference frame $n$  is the normalized energy density $T_{\mu\nu} n^\mu n^\nu$. This choice as the apparent drawback that probability densities of different reference frames result to  be incomparable, just because the densities $T_{\mu\nu} n^\mu n^\nu$ and $T_{\mu\nu} n'^\mu n'^\nu$ are not connected by the standard argument based on the conservation law $\partial_\mu T^{\mu\nu}=0$ and the Stokes-Poincar\'e theorem. That law permits to compare different boundary terms where only one normal vector is changed instead of one pair at a time: $n,n \to n',n'$. To test Castrigiano's  causality condition seems to be impossible along that way. However,  the physical interpretation turns out to be of some help at this juncture. The twice presence of $n$ can be relaxed to  a single occurence of a pair of different timelike future-oriented  unit vectors, $n,n'$. The fact that the density $T^{\mu\nu} n_\mu n_\nu'$ is still positive suggests  a new and different  operational interpretation of the  notion  of spatial position.   
To assert that the particle stays in $\Delta \subset \Sigma_{n',t'}$ one should not only specify the reference frame $n'$ and the instant of time $t'$, but one should also make explicit our choice of the rest frame $n$ of the employed  detectors (which actually are energy detectors). 
The relevant density therefore is $J_{n}^\mu n'_\mu \geq 0$, where $J_{n}^\mu := n^\nu T_{\nu}^\mu$.
This picture includes  the  apparently most natural choice is $n= n'$, but one  is also allowed to pick out  $n\neq n'$.  Keeping fixed $n$ and varying $n'$ produces   a new family of POVMs $\sM^n_{n',t'}$ when one varies $n'$ and $t'$. This family satisfies both requirements  (a) and (b) in Definition \ref{LCR}, in particular Castigiano's causality condition (b). It is not clear to the author if this approach is really physically  meaningful and the subject  certainly deserves further investigation and discussion.

Actually something can be said about the causal relation of $\sA_{n,t}(\Delta_1)$ and $\sA_{n',t'}(\Delta_2)$, where $\Delta_2 = (J^+(\Delta_1) \cup J^-(\Delta_1)) \cap \Sigma_{n',t'}$
and $n\neq n'$,  on the ground of a pure mathematical observation. However,  it is not clear if this reasoning   may lead to a proof of  Castrigano's causality condition, especially because there is no evident physical reason behind the following argument. If one assumes that $\psi \in {\cal D}({\cal H})$, and that $\Delta_1\subset \Sigma_{n,t_1}$ has the special form as in  Proposition \ref{PROPSHARP}, then the sharp inequality (\ref{sharp}) is valid. Therefore,  for continuity reasons, keeping fixed $\psi$,  $n$ and $t=t_1$ on the left-hand side of  (\ref{sharp}), that inequality must be still valid if one slightly changes $n'=n$ and $t'= t_2$,  and $\Delta_2$ accordingly. If the neighborhood of values  $(n',t')$ around $(n,t)$  where this inequality holds  were the entire  $\sT_+\times \bR$, one could  use an improvement of the argument already exploited in the main text to pass from the  special type of set $\Delta_1$ to a generic element of $\cL(\Sigma_{n,t})$, possibly  relaxing  $<$ to  $\leq$. The usual density argument of ${\cal D}({\cal H})$ in ${\cal H}$ would conclude the proof.  However, I do not think that the said neighborhood of $(n,t)$ covers the full set of possibilities of the choice of $(n',t')$. All that will be investigated elsewhere.

\section*{Acknowledgments} I am very grateful to D.P.L.Castrigiano  for various remarks, suggestions, and discussions about several issues appearing  in this paper. I thank  S.Delladio, N.Drago, C.Fewster, F.Finster,  S.Mazzucchi, P.Meda,  M.Sanch\'ez for helpful discussions.  I am finally grateful to a referee for very helpful comments a suggestions of various nature, including further relevant  references. This work has been written within the activities of INdAM-GNFM

\appendix
\section{Proof of some propositions}\label{APPPROOFS}

\noindent {\bf Proof of Proposition \ref{PROPS}}. The first two  statements are  evident per direct inspection.  The density property arises from the fact that   the Schwartz space $\cS(\bR^3)$ is dense in $L^2(\bR^3, d^3p)$. Therefore, if $\psi \in {\cal H}$, there is a sequence $\cS(\bR^3) \ni \psi_n$ with $$\int_{\bR^3} \left|\frac{\psi(\vec{p}_n)}{E_n(\vec{p}_n)} - \psi_n(\vec{p}_n)\right|^2 d^3p \to 0\quad \mbox{as $n\to +\infty$.}$$ However 
$\psi' :=\sqrt{ E_n} \psi \in \cS(\bR^3)$ as well, and  $\int_{\bR^3} \left|\psi(\vec{p}_n) - \psi'_n(\vec{p}_n)\right|^2 \frac{d^3p}{E_n(\vec{p}_n)} \to 0$. The sequence of $\psi'_n$ belongs to 
${\cal S}({\cal H})$ by definition and converges to $\psi$ in the topology of ${\cal H}$ so that the thesis is true. \hfill $\Box$\\

\noindent {\bf Proof of Proposition \ref{PA}}.
 The dense subspace  ${\cal S}({\cal H})$ stays in the domains of the considered operators, it is invariant and thereon the operators are symmetric.
The multiplicative action of the one-parameter groups generated by the said four operators leaves   ${\cal S}({\cal H})$ invariant, as it arises per direct inspection. As a consequence of a known corollary of the Stone theorem  (see, e.g., Corollary 7.26 in \cite{Moretti2}) the thesis follows. \hfill $\Box$\\

\noindent {\bf Proof of Proposition \ref{propX}}.
 First observe that $ N^0_{n,t}$ is nothing but $tI$ so that (1) and   (2) are trivial for it. 
Assuming $t=0$, let us focus again on the unitary map (\ref{mapV})
$$ S_n: L^2(\sV_{m,+}, \mu_m) \ni \psi(p) \mapsto \frac{\psi(E_n(p), \vec{p}_n)}{\sqrt{E_n(p)}} \in L^2(\bR^3, d^3p)\quad \mbox{such that 
$\cS(\bR^3) = S_n({\cal S}({\cal H}))$.}$$
 Per direct inspection one sees that $P'_{n \alpha} :=S_n P_{n \alpha} S_n^{-1} $ is still a multiplicative operator
 $\vec{p}_{nk} \cdot$ (for $k = 1,2,3$)
  in $ L^2(\bR^3, d^3p)$. Similarly, from (\ref{CONVNW2}), ${ N'}_{n,0}^{k}:=S_n  N^k_{n 0} S_n^{-1}$ is the (selfadjoint)  multiplicative operator $x^k\cdot $ in $L^2(\bR^3,d^3x)$, where $L^2(\bR^3,d^3x)$ and $L^2(\bR^3, d^3p)$ are connected to each other by the Fourier-Plancherel unitary transform. So that  these sets of operators are exactly the non-relativistic ones in $L^2(\bR^3,d^3p)$ and $L^2(\bR^3, d^3x)$. As a consequence, (1), (2), and  (3) are valid because they are valid for the non relativistic operators if replacing ${\cal S}({\cal H})$
for $\cS(\bR^3) = S_n({\cal S}({\cal H}))$ (e.g., see \cite{Moretti2}) 
 and the considered properties are invariant under unitary maps. If we switch on $t\neq 0$, since $ N^\alpha_{n,t} = U^{(n)-1}_t  N^\alpha_{n,t} U^{(n)}_t$
and $P_{n \alpha} =  U^{(n)-1}_t P_{n \alpha} U^{(n)}_t$  as a consequence of the analogs for the corresponding spectral measures, the found properties are still valid because 
 the evolutor $U_t^{(n)}$ is unitary and leaves ${\cal S}({\cal H})$ invariant. Let us pass to the proof of (5). From (\ref{COVNW}),  $D( N_{n,t}^\alpha) \supset {\cal S}({\cal H})$, and the definition (\ref{NWx}), we have
$$\langle  \psi'| U_h  N_{n,t}^\alpha U_h^{-1} \psi \rangle =   \int_{x\in \Sigma_{n,t}} x^\alpha d\langle \psi'| 
\sQ_{\Lambda_h n,t_h}(hx) \psi \rangle = \int_{hx\in \Sigma_{\Lambda_hn,t_h}} x^\alpha d\langle \psi'| 
\sQ_{\Lambda_h n,t_h}(hx) \psi \rangle$$
$$=\int_{hx\in \Sigma_{\Lambda_hn,t_h}} (h^{-1}hx)^\alpha d\langle \psi'| 
\sQ_{\Lambda_h n,t_h}(hx) \psi \rangle =  \int_{y\in \Sigma_{\Lambda_hn,t_h}} ((\Lambda^{-1}_h)(y-a_h))^\alpha d\langle \psi'| 
\sQ_{\Lambda_h n,t_h}(y) \psi \rangle$$
where $\psi'\in {\cal H}$ and $\psi \in {\cal S}({\cal H})$.
The last integral equals $$\langle \psi' |  (\Lambda^{-1}_h)^\alpha_\beta ( N^{\beta}_{\Lambda_h n, t_{h}} - a_h^\beta I)\psi \rangle\:,$$
which implies the thesis due to  arbitariness of $\psi' \in {\cal H}$. Only (4), i.e., the pair of identities in (\ref{EV2}), remain to be proved for $\alpha=k=1,2,3$.  The first identity 
$U^{(n)\dagger}_t  N_{n,0}^kU^{(n)}_t\psi  =   N_{n,t}^k\psi $ for  $\psi \in {\cal S}({\cal H})$ immediately arises from (\ref{55}). Let us pass to the second identity in (\ref{EV2}).
Define $f(\vec{p}_n) := (S_n\psi)(\vec{p}_n)$ where $S_n$ is the unitary map (\ref{mapV}). The operators  $P_{nk}$ and $H_n$ acts on the functions 
$f=f(\vec{p}_n)$ multiplicatively, respectively with $p_k$ and $\sqrt{\vec{p}^2+m^2}$, whereas $ N^k_{n,0}$ is represented by $i\frac{\partial}{\partial p_k}$, finally $U^{(n)}_t$ is the multiplicative operator with $e^{-it \sqrt{\vec{p}^2+m^2}}$. 
As a consequence, for $\psi,\psi' \in {\cal S}({\cal H})$ (writing $\vec{p}$ in place of $\vec{p}_n$)
$$\langle \psi'|  N^k_{n,t} \psi \rangle =\langle U^{(n)}_t\psi'|  N^k_{n,0} U^{(n)}_{t}\psi \rangle= \int_{\bR^3} d^3p e^{it \sqrt{\vec{p}^2+m^2}} \overline{f'(\vec{p})}
i\frac{\partial}{\partial p_k}e^{-it \sqrt{\vec{p}^2+m^2}} f(\vec{p})$$
where $f= S_n(\psi) \in \cS(\bR^3)$ and $f'= S_n(\psi') \in \cS(\bR^3)$. Using the fact that $f$ and $f'$ are Schwartz, the $t$-derivative of the integral above can be computed by passing the derivative under the sign of integral (by a straightforward use of Lebesgue's dominated convergence theorem) finding
 $$\frac{d}{dt}\langle \psi'|  N^k_{n,t} \psi \rangle =\langle U^{(n)}_t\psi'|  N^k_{n,0} U^{(n)}_{t}\psi \rangle=
i^2  \int_{\bR^3}\sp\sp  d^3p e^{it \sqrt{\vec{p}^2+m^2}} \overline{f'(\vec{p})}
\left[ \sqrt{\vec{p}^2+m^2}, \frac{\partial}{\partial p_k} \right]e^{-it \sqrt{\vec{p}^2+m^2}} f(\vec{p})
  $$
  $$=  \int_{\bR^3}\sp\sp  d^3p e^{it \sqrt{\vec{p}^2+m^2}} \overline{f'(\vec{p})}
 \frac{p_k}{\sqrt{\vec{p}^2 + m^2}}  e^{-it \sqrt{\vec{p}^2+m^2}} f(\vec{p})
 =  \int_{\bR^3}\sp\sp  d^3p \overline{f'(\vec{p})}
 \frac{p_k}{\sqrt{\vec{p}^2 + m^2}} f(\vec{p}) = \langle \psi'| H_n^{-1} P_{nk} \psi \rangle\:.$$
 As the final result does not depend on time, we can argue that
 $$\langle \psi'|  N^k_{n,t} \psi \rangle = \langle \psi'|  N^k_{n,0} \psi \rangle + t  \langle \psi'| H_n^{-1} P_{nk} \psi \rangle\:.$$
 Namely,
  $$\langle \psi'| ( N^k_{n,t} -  N^k_{n,0} - t   H_n^{-1} P_{nk}) \psi \rangle =0\:.$$
 Since $\psi'\in {\cal S}({\cal H})$ which is dense, the found result implies the thesis.
\hfill $\Box$\\

\noindent {\bf Proof of Corollary \ref{CORR}}.
 We shall write $P_k$  in place of  $P_{nk}$ and $H$ in place of $H_n$ for shortness. As $\psi \in {\cal S}({\cal H})$ which is invariant under $P_k$ and $H$,  no domain issues take place in the following. Due to (\ref{EV}), the thesis is equivalent to
 $$\sum_{k=1}^3 \langle \psi|H^{-1}P_{k} \psi\rangle^2< 1\:. $$
 To prove it, observe that $H^{-1} P_k$ is well defined and symmetric on ${\cal S}({\cal H})$, hence
 $$\langle \psi |(H^{-1} P_k) (H^{-1}P_k)\psi \rangle - \langle \psi|H^{-1} P_k \psi\rangle^2 = 
\langle \psi| (H^{-1} P_k - \langle \psi|  (H^{-1} P_k) \psi \rangle I)^2 \psi \rangle
  \geq 0$$ so that, since $(H^{-1} P_k) (H^{-1}P_k)\psi  = H^{-2}P^2_k\psi$ for $\psi \in {\cal S}({\cal H})$,
  $$\langle \psi |H^{-2} P^2_k\psi \rangle\geq  \langle \psi|H^{-1} P_k \psi\rangle^2\:.$$
  As a consequence
  $$1 = \langle \psi|\psi\rangle = \sum_{k=1}^3\langle \psi |H^{-2} P^2_k\psi \rangle+ m^2 \langle \psi |H^{-2}\psi \rangle \geq  \sum_{k=1}^3 \langle \psi|H^{-1} P_k \psi\rangle^2
  +  m^2 \langle \psi |H^{-2}\psi \rangle\:. $$
  Since $m^2\langle \psi|H^{-2} \psi\rangle = m^2|| H^{-1}\psi||^2> 0$  ($H^{-1}\psi=0$ is not possible if $\psi \neq 0$ because, as $H^{-1}: {\cal H}= Ran(H) \to D(H)$, it would imply $0=HH^{-1}\psi = \psi$), the inequality above implies the thesis. \hfill $\Box$\\

\noindent {\bf Proof of Eq.(\ref{limit})}. From (\ref{TPOVM}) and the definition of $\psi_j$, 
 $$\left\langle \psi_j \left|\left( \eta^{\mu\nu}\frac{P_{n\mu}}{H_n}  \sQ_{n,0}(B_R) \frac{P_{n\nu}}{H_n}  + \frac{m}{H_n} \sQ_{n,0}(B_R)  \frac{m}{H_n} \right) \right.\psi_j \right\rangle\:,$$
up to a non-vanishing  multiplicative constant, coincides with
$$I_j = \int_{\bR^6} d^3p d^3 q  \:\overline{\hat{\chi}(\vec{q} -j \vec{a})}\hat{\chi}(\vec{p}-j\vec{a}) f(|\vec{p}-\vec{q}\:|) \frac{p\cdot q + m^2}{E_n(p)E_n(q)}$$
where  $\hat{\chi}$ is a Schwartz function on $\bR^3$ and  $f$ is the Fourier transform (up to a constant factor) of the characteritic function of $B_R$,
$$f(|\vec{p}-\vec{q}|) = \int_0^R \frac{\sin (r |\vec{p}-\vec{q}|)}{|\vec{p}-\vec{q}|}  r dr= \frac{\sin (R |\vec{p}-\vec{q}|)  - R|\vec{p}-\vec{q}| \cos (R|\vec{p}-\vec{q}|)}{|\vec{p}-\vec{q}|^3}\:.$$
Since 
$$\frac{p\cdot q + m^2}{E_n(p)E_n(q)} = \frac{|\vec{p}-\vec{q}|^2((\vec{p}+\vec{q})^2 - (E_n(p)+E_n(q))^2)}{2 E_n(p)E_n(q)(E_n(p)+E_n(q))^2}\:,$$
we have 
$$I_j = \int_{\bR^6} d^3p d^3 q  \:\overline{\hat{\chi}(\vec{q} -j \vec{a})}\hat{\chi}(\vec{p}-j\vec{a})\frac{\sin (R |\vec{p}-\vec{q}|)  - R|\vec{p}-\vec{q}| \cos (R|\vec{p}-\vec{q}|)}{2|\vec{p}-\vec{q}|  E_n(p)E_n(q)}  \frac{((\vec{p}+\vec{q})^2 - (E_n(p)+E_n(q))^2)}{(E_n(p)+E_n(q))^2}\:. $$
Using the fact that the last factor, $\cos u$, and  $u^{-1} \sin u$ are bounded, we have that, for some $C\geq 0$,
$$|I_j| \leq C \int_{\bR^6} d^3p d^3 q \frac{ \:|\overline{\hat{\chi}(\vec{q} -j \vec{a})}||\hat{\chi}(\vec{p}-j\vec{a})|}{ E_n(p)E_n(q)}  
\leq \left(\int_{\bR^3} d^3q\frac{|\hat{\chi}(\vec{q} -j \vec{a})|^{4/2}}{E_n(q)^4}  \right)^{1/4}
\left(\int_{\bR^3} d^3q|\hat{\chi}(\vec{q} -j \vec{a})|^{1/2 \cdot 4/3} \right)^{3/4}$$
$$\left(\int_{\bR^3} d^3p\frac{|\hat{\chi}(\vec{p} -j \vec{a})|^{4/2}}{E_n(p)^4}  \right)^{1/4}
\left(\int_{\bR^3} d^3p|\hat{\chi}(\vec{p} -j \vec{a})|^{1/2 \cdot 4/3} \right)^{3/4}$$
where we have used H\"older's inequality in the last passage. As a matter of fact, since the Lebesgue measure is translationally invariant, there is  $K\geq 0$ such that, uniformly in $j$,
$$|I_j| \leq K \left(\int_{\bR^3} d^3p\frac{|\hat{\chi}(\vec{p} -j \vec{a})|^{2}}{E_n(p)^4}  \right)^{1/2}\:.$$
The integrand is $j$-uniformly bounded by the integrable function $\frac{K'}{E_n(p)^4}$ for some constant $K'\geq 0$ and the integrand vanishes pointwise as $j\to +\infty$ as $\hat{\chi} \in \cS(\bR^3)$. Lebesgue's  dominated convergence theorem implies that $I_j \to 0$ as $j\to +\infty$.  \hfill $\Box$\\

\noindent {\bf Proof of Proposition} \ref{PROPSHARP}. We start where the proof of Lemma \ref{LEMMAA} ends, with the further hypothesis that $\psi \in {\cal D}({\cal H})$. We first consider the case of $\Delta_1$ made of a single ball.   Since $-J^v\geq 0$ is continuous, the integral in (\ref{intinT}) vanishes if and only if $J^v=0$ everywhere on $L$. This is the only possibility for having $\langle \psi | \sA_{n,t_1}(\Delta_1) \psi\rangle = \langle \psi | \sA_{n,t_2}(\Delta_2) \psi\rangle$. Let us prove that $J^v=0$ everywhere in $L$ is not permitted and this fact will conclude the proof. Let us assume that $J^v=0$ on $L$ so that $J_n^\psi$ vanishes or is lightlike on $L$  because  $-J^uJ^v + h(\vec{J}, \vec{J}) \leq 0$ and the only remaining component is $J^u$.
 From Proposition \ref{PROPBAST1} we know that $\Phi^\psi_n(x)=0$ if $x\in L$. Making explicit the form of $\Phi^\psi_n$ on $L$, in terms our coordinate system, we have that
 $$\Phi^{\psi}_n(t, r, \theta, \phi ) =\int_{\sV_{m,+}}  \frac{\psi(p)e^{i |\vec{p}_n|r \cos \alpha - i E_n(\vec{p}_n)t}}{(2\pi)^{3/2}\sqrt{E_n(\vec{p}_n)}} \frac{d^3p}{E_n(\vec{p}_n)} \:,
 $$
 where $$\cos \alpha = \sin\theta \sin \theta_p \cos(\phi-\phi_p) + \cos\theta \cos\theta_p$$ and $\theta_p,\phi_p$ are the polar angles of  $\vec{p}_n$.
Passing to lightlike coordinates and noticing that $L$ is described by $v=0$, we have in particular that it must be
 $$0= \Phi^{\psi}_n(u, v=0, \theta, \phi ) =\int_{\sV_{m,+}}  \frac{\psi(p)e^{i |\vec{p}_n|\frac{u}{2} \cos \alpha - i E_n(\vec{p}_n) \frac{u}{2}}}{(2\pi)^{3/2}\sqrt{E_n(\vec{p}_n)}} \frac{d^3p}{E_n(\vec{p}_n)} \quad u \in [a,b]\:, \theta \in [0,\pi]\:, \phi \in [-\pi,\pi]$$
 where $a<b$ are determined by $t_2-t_1$ and the radius of $\Delta_1$.
Since $\psi$ is continuous with compact support (here the condition $\psi \in {\cal D}({\cal H})$ is used), by a standard argument based on the Cauchy-Riemann identities and the Lebesgue dominated convergence theorem  it is easy to prove that the function in the right-hand side can be analytically extended to complex values of $u$ in the whole complex plane.
 As this function vanishes in the real segment $[a,b]$, it must vanish everywhere in $u \in [0,+\infty)$. \\
We observe for future convenience that the same argument can be used to prove that the integral is an analytic function in the variables $\theta$ and $\phi$ and that if the function vanishes in an open interval in the domain of $\theta$ or in an analogous  open interval  in the domain of  $\phi$, then it must vanish for all the permitted values of these variables, respectively,  $\theta \in  [0,\pi]$ 
and $\phi \in [-\pi,\pi]$. To assert that $\Phi^\psi_n=0$ on the whole conical surface described by $u\in [0,+\infty)$, $\theta \in [0,\pi]$, $\phi\in [-\pi,\pi]$
 it is therefore sufficient that $\Phi^\psi_n=0$ on an open set on that conical surface.\\
The conclusion is that the smooth solution $\Phi^\psi_n$  of the massive Klein-Gordon equation in  $\bM$  vanishes on the whole conical surface defined  by prolonging $\partial J^+(D_1)$ for times $<t_1$ up to the tip of the cone. As is known \cite{Hoermander, Friedlander}, the {\em characteristic Cauchy problem} (also known as the  {\em Goursat problem}) is well-posed inside a Lorentzian cone and thus the only possible solution inside the volume of the cone is $\Phi^\psi_n=0$. In other words our wavefunction, defined in the whole $\bM$ must vanishes in the volume of the cone. In particular, $\Phi_n^\psi(t_1,\cdot)$ and $i\partial_t \Phi_n^\psi(t_1,\cdot) = (\overline{-\Delta +m^2I})^{1/2} \Phi_n^\psi(t_1,\cdot)=0$ in the open ball $\Delta_1$. Theorem \ref{teorem11} implies that it vanishes on the whole $\Sigma_{n,t_1}$.  Inverting (\ref{wavePHI}), we have $\psi=0$ that is not possible since $||\psi||=1$ by hypothesis. The hypothesis $J^v=0$ everywhere on $L$ is untenable and this fact removes the possibility of having $=$ in (\ref{INEQJ}) proving the thesis for the considered case.\\
Let us pass to consider the case of  $\Delta_1 = \Delta^{(1)}_1 \cup \Delta^{(2)}_2$ with the two sets being a pair of  non-empty finite-radius  open balls. 
We can always assume that each ball does not include the other but they can have non-empty intersection.
We have 
\beq \langle \psi | \sA_{n,t_2}(\Delta_2) \psi\rangle - \langle \psi | \sA_{n,t_1}(\Delta_1) \psi\rangle = \int_{L_{12}} \nu^\psi_n\label{DIFF}\eeq
where $L_{12}$ is the part of  $\partial J^+(\Delta^{(1)}_1 \cup \Delta^{(2)})$ which stays between the parallel planes $\Sigma_{n,t_1}$ and $\Sigma_{n,t_2}$.
As before the integral is non-negative because we can apply the previous argument to each portion of conical surface
forming $L_{12}$ and respectively generated by  $\Delta_1^{(1)}$ and $\Delta_2^{(2)}$, taking advantage of two different polar coordinate systems.
 However, the fact that the integral is strictly positive needs a little more care. As before, on account of  Proposition \ref{PROPBAST1},  the value of  integral is zero if and only if $\Phi^\psi_n$ everywhere\footnote{The singular regions of the set $\partial J^+(\Delta^{(1)}_1 \cup \Delta^{(2)})$ where the set ceases to be an embedded submanifold are reached by continuity of $\Phi^\psi_n$.} vanishes on $L_{12}$. We can focus attention on the complete conical surface $\Gamma_1$  which completes  $\partial J^+(\Delta^{(1)}_{1})$ in its past till the tip, centering a system of polar coordinates on its center. It is clear that the intersection of  $\Gamma_1 \cap L_{12}$ includes  an open set (in the relative topology of $\Gamma_1$)  where $\Phi^\psi_n$ vanishes because  it vanishes on the whole $L_{12}$. Using the analyticity argument exploited above, we conclude that $\Phi^\psi_n$ vanishes on the whole $\Gamma_1$, so that it also vanishes in the interior of the cone in view of the characteristic  Cauchy problem  as before, and finally $\Phi^\psi_n=0$ everywhere in $\bM$ due to  Theorem \ref{teorem11} reaching a contradiction $\psi=0$. Hence the right-hand side of (\ref{DIFF}) is strictly positive and the proof for the examined case is over.\\
 To conclude the proof it is sufficient to observe what follows  in the case  $\Delta_1$ is a finite union of distinct finite-radius open balls $\Delta^{(j)}_{1}$, 
 $j=1,\ldots, N$. We can always assume that  no  ball of the family is a subset of another ball of the family.
  Since $N$ is finite,   the region of  $\partial J^+(\Delta_1)$ 
 between $t_1$ and $t_2$  necessarily includes  an open portion of some  $\partial J^+(\Delta^{(j)}_1)$. Working in the conical completion $\Gamma_j$ of  $\partial J^+(\Delta^{(j)}_1)$, we can use the above argument achieving the thesis.


\begin{thebibliography}{999}
\bibitem{Akniezer-Glazman-93}
	N.I. Akhiezer, I.M. Glazman,
	\textit{Theory of Linear Operators in Hilbert Space},
	Dover, NewYork, 1993, two volumes bound as one, translated from Russian by M. Nestell, first published by F. Ungar, NewYork, in 1961 and 1963.


\bibitem{tesi} C. Beck, {\em Localization
Local Quantum Measurement and Relativity}, Dissertation
an der Fakult\"at f\"ur Mathematik, Informatik und Statistik
der Ludwig-Maximilian-Universit\"at M\"unchen, 2020

\bibitem{BK} N. Barat and J. C. Kimball, {\em Localization and Causality for a free particle} Phys. Lett. A 308, 110 (2003)

\bibitem{BB} I. Bialynicki-Birula and Z. Bialynicka-Birula, {\em Heisenberg uncertainty relations for photons}, Phys. Rev. A 86,
022118 (2012)

\bibitem{BFR} H. Bostelmann, C. J. Fewster, and M. H. Ruep, {\em Impossible measurements require impossible apparatus}, Phys. Rev. D 103, 025017 (2021)



\bibitem{Buschloc} P.Busch, {\em Unsharp localization and causality in relativistic quantum theory}. J.
Phys. A: Math. Gen. 32, 37 (1999), 6535


\bibitem{Buschbook} P.Busch , P. Lahti , J.-P. Pellonp\"a\"a, K. Ylinen, {\em Quantum Measuremement}, Springer 2016




\bibitem{POVMEUC} C. Carmeli, G. Cassinelli, E. De Vito, A. Toigo, and B Vacchini, {\em A complete characterization of phase space
measurements}, J. Phys. A: Math. Gen. 37 (2004) 5057-5066

\bibitem{Castrigiano2}     D. P. L. Castrigiano, {\em Dirac and Weyl Fermions - the Only Causal Systems} 2017,  arXiv:1711.06556


\bibitem{Castrigiano1}   D.P.L. Castrigiano, A.D. Leiseifer, {\em Causal Localizations in Relativistic Quantum Mechanics},  J. Math.
Phys. 56, 072301 (2015)


\bibitem{Cohn} D. Cohn, {\em Measure Theory}. Birkh\"auser (1980)

\bibitem{DM}  N. Drago and V.Moretti, {\em  The notion of observable and the moment problem for $*$-algebras and their GNS representations}, 
Lett. Math. Phys. 110(7), (2020), 1711-1758




\bibitem{EG} L. C. Evans  and  R. F. Gariepy ,  {\em Measure Theory and Fine Properties of Functions} Revised Edition, CRC Press
Taylor \& Francis Group (2015)




\bibitem{FewsterVerch} C. J. Fewster and R Verch, {\em Quantum Fields and Local Measurements}, Commun. Math. Phys.  {\bf 378},  851-889 (2020)



\bibitem{FJR} C.L. Fewster, I. Jubb, I., and M.H.  Ruep, {\em Asymptotic Measurement Schemes for Every Observable of a Quantum Field Theory}. Ann. Henri Poincar\'e (2022)

\bibitem{FKW} Sz. Farkas, Z. Kurucz, and M. Weiner, {\em Poincar\'e Covariance of Relativistic Quantum Position}
, International Journal of Theoretical Physics, Vol. 41, No. 1, January 2002

\bibitem{Friedlander} F.G. Friedlander,  {\em The Wave Equation on a Curved Space-Time}, Cambridge University Press (1976)


\bibitem{A} B. Gerlach, D. Gromes, J. Petzold, {\em Konstruktion definiter Ausdr\"ucke f\"ur die Teilchendichte des Klein-Gordon-Feldes} Z. Physik, Volume 204, Issue 1, pp.1-11 (1967)

\bibitem{B}  B. Gerlach, D. Gromes, J. Petzold, P. Rosenthal {\em \"Uber kausales Verhalten nichtlokaler Gr\"ossen und Teilchenstruktur in der Feldtheorie},
 Z. Physik 208, 381-389 (1968)


\bibitem{HC} H. Halvorson and R. Clifton, {\em No place for particles in relativistic quantum theories?} in Ontological Aspects of Quantum Field Theory, Edited By: M. Kuhlmann, H. Lyre, and A. Wayne, World Scientific November 2002

\bibitem{Hegerfeldt} G. C. Hegerfeldt,  {\em Remark on causality and particle localization},  Phys. Rev. D 10, 3320  (1974)
\bibitem{Hegerfeldt2} G. C. Hegerfeldt, {\em Violation of Causality in Relativistic Quantum Theory?}, Phys. Rev. Lett. 54, 2395 (1985)


\bibitem{HK} K.E. Hellwig and K.  Kraus,  {\em Formal Description of Measurements in Local Quantum Field Theory}, Phys. Rev. D 1970

\bibitem{C} J.J. Henning, W. Wolf, {\em Positive definite densities for the positive frequency solutions of the Klein-Gordon equation with arbitrary mass}, Z. Phys. 242, 12-20 (1971) 


\bibitem{Hoermander} L. H\"ormander, {\em A remark on the characteristic Cauchy problem}, J. Funct. Anal. 93 (1990),
270-277


\bibitem{Jancewicz} B. Jancewicz, {\em Operator density current and relativistic localization problem}, J. Math. Phys. 18, 2487 (1977)


\bibitem{D} M.J. Kazemi, H. Hashamipour, M.H. Barati, {\em Probability density of relativistic spinless particles}, Phys. Rev. A 98, 012125 (2018) 

\bibitem{Malament} D.B. Malament, {\em In Defense of  Dogma: Why There  Cannot  Be a
Relativistic Quantum  Mechanics  of  (Localizable) Particles} in R. Clifton (ed.).  {\em Perspectives on Quantum Reality}, 1-10.
e 1996 Kluwer Acodemic Publishers. Printed in the Netherlands

\bibitem{Moretti1}
	V. Moretti,
	\textit{Spectral Theory and Quantum Mechanics}, 2nd revised and enlarged
edition,
	Springer International Publishing (2017)
\bibitem{Moretti2}
	V. Moretti,
	\textit{Fundamental Mathematical Structures of Quantum Theory},
	Springer International Publishing (2019)

\bibitem{murata}  M. Murata. {\em Anti-locality of certain functions of the Laplace operator}.   J. Math. Soc. Japan 25(4): 556-564 (October, 1973)

\bibitem{Naimark1940}  N.Naimark,  \textit{Self-adjoint extensions of the second kind of a symmetric operator}, Izv. A.N. SSSR, {\bf 4}, 1, 53-104 (1940)
\bibitem{Naimark1940b}  N.Naimark,  \textit{Spectral functions of a symmetric operator}, Izv. A.N. SSSR, {\bf 4}, 3, 277-318 (1940)

 \bibitem{NW} T.D. Newton, E.P. Wigner,  {\em Localized States for Elementary Systems}, Rev. Mod. Phys. 21, 400-406
(1949)

\bibitem{Ozawa} M. Ozawa, {\em Quantum measuring processes of continuous observables}, J. Math. Phys. 25, 79 (1984)

\bibitem{Rudin} W. Rudin, {\em Real and Complex Analysis} 3d edition, McGraw-Hill, (1986)

 \bibitem{Ruijsenaars} S.N.M. Ruijsenaars, {\em On Newton-Wigner Localization and Superluminal Propagation Speeds}, Ann. Phys. 137, 33-43 (1981)

\bibitem{Terno} D. R. Terno,  {\em Localization of relativistic particles and uncertainty relations}, Phys. Rev. A 89, 042111 (2014)

\bibitem{anti1} E. Segal and R. W. Goodman, {\em Anti-locality of certain Lorentz-invariant operators}, J. Math. and Mech. 14, No. 4 (1965), 629-638



\bibitem{WM} W. Weidlich, A.K. Mitra, {\em Some Remarks on the Position Operator in Irreducible Representations of the Lorentz-Group.} Nuovo Cim. 30, 385-389 (1963)

\bibitem{Wightman} A.S. Wightman, {\em On the Localizability of Quantum Mechanical Systems}, Rev. Mod. Phys. 34, 845-872
(1962)

\end{thebibliography}
\end{document}